\documentclass[pra,superscriptaddress,showpacs,onecolumn,twoside,11pt,notitlepage,a4paper]{revtex4-1}

\usepackage{graphicx,epic,eepic,epsfig,amsmath,latexsym,amssymb,verbatim,color}

\usepackage{theorem}
\newtheorem{definition}{Definition}
\newtheorem{proposition}[definition]{Proposition}
\newtheorem{lemma}[definition]{Lemma}

\newtheorem{theorem}[definition]{Theorem}

\def\squareforqed{\hbox{\rlap{$\sqcap$}$\sqcup$}}
\def\qed{\ifmmode\squareforqed\else{\unskip\nobreak\hfil
\penalty50\hskip1em\null\nobreak\hfil\squareforqed
\parfillskip=0pt\finalhyphendemerits=0\endgraf}\fi}
\def\endenv{\ifmmode\;\else{\unskip\nobreak\hfil
\penalty50\hskip1em\null\nobreak\hfil\;
\parfillskip=0pt\finalhyphendemerits=0\endgraf}\fi}
\newenvironment{proof}{\noindent \textbf{{Proof~} }}{\qed}
\newenvironment{remark}{\noindent \textbf{{Remark~}}}{}

\mathchardef\ordinarycolon\mathcode`\:
\mathcode`\:=\string"8000
\def\vcentcolon{\mathrel{\mathop\ordinarycolon}}
\begingroup \catcode`\:=\active
  \lowercase{\endgroup
  \let :\vcentcolon
  }

\newcommand{\nc}{\newcommand}
\nc{\rnc}{\renewcommand}
\nc{\beq}{\begin{equation}}
\nc{\eeq}{\end{equation}}
\nc{\beqa}{\begin{eqnarray}}
\nc{\eeqa}{\end{eqnarray}}
\nc{\lbar}[1]{\overline{#1}}
\nc{\bra}[1]{\langle#1|}
\nc{\ket}[1]{|#1\rangle}
\nc{\ketbra}[2]{|#1\rangle\!\langle#2|}
\nc{\braket}[2]{\langle#1|#2\rangle}
\nc{\proj}[1]{| #1\rangle\!\langle #1 |}
\nc{\avg}[1]{\langle#1\rangle}
\nc{\Rank}{\operatorname{Rank}}
\nc{\smfrac}[2]{\mbox{$\frac{#1}{#2}$}}
\nc{\tr}{\operatorname{Tr}}
\nc{\ox}{\otimes}
\nc{\dg}{\dagger}
\nc{\dn}{\downarrow}
\nc{\cA}{{\cal A}}
\nc{\cB}{{\cal B}}
\nc{\cC}{{\cal C}}
\nc{\cD}{{\cal D}}
\nc{\cE}{{\cal E}}
\nc{\cF}{{\cal F}}
\nc{\cG}{{\cal G}}
\nc{\cH}{{\cal H}}
\nc{\cI}{{\cal I}}
\nc{\cJ}{{\cal J}}
\nc{\cK}{{\cal K}}
\nc{\cL}{{\cal L}}
\nc{\cM}{{\cal M}}
\nc{\cN}{{\cal N}}
\nc{\cO}{{\cal O}}
\nc{\cP}{{\cal P}}
\nc{\cQ}{{\cal Q}}
\nc{\cR}{{\cal R}}
\nc{\cS}{{\cal S}}
\nc{\cT}{{\cal T}}
\nc{\cX}{{\cal X}}
\nc{\cY}{{\cal Y}}
\nc{\cZ}{{\cal Z}}
\nc{\csupp}{{\operatorname{csupp}}}
\nc{\qsupp}{{\operatorname{qsupp}}}
\nc{\var}{{\operatorname{var}}}
\nc{\rar}{\rightarrow}
\nc{\lrar}{\longrightarrow}
\nc{\polylog}{{\operatorname{polylog}}}
\nc{\1}{{\openone}}
\nc{\wt}{{\operatorname{wt}}}
\nc{\av}[1]{{\left\langle {#1} \right\rangle}}

\nc{\RR}{{{\mathbb R}}}
\nc{\CC}{{{\mathbb C}}}
\nc{\FF}{{{\mathbb F}}}
\nc{\NN}{{{\mathbb N}}}
\nc{\ZZ}{{{\mathbb Z}}}
\nc{\PP}{{{\mathbb P}}}
\nc{\QQ}{{{\mathbb Q}}}
\nc{\UU}{{{\mathbb U}}}
\nc{\EE}{{{\mathbb E}}}
\nc{\id}{{\operatorname{id}}}

\nc{\CHSH}{{\operatorname{CHSH}}}

\nc{\Aram}{{\operatorname{\sf A}}}

\nc{\be}{\begin{equation}}
\nc{\ee}{{\end{equation}}}
\nc{\bea}{\begin{eqnarray}}
\nc{\eea}{\end{eqnarray}}
\nc{\<}{\langle}
\rnc{\>}{\rangle}
\nc{\Hom}[2]{\mbox{Hom}(\CC^{#1},\CC^{#2})}
\nc{\rU}{\mbox{U}}

\nc{\ob}[1]{#1}

\begin{document}

\title{On zero-error communication via quantum channels\protect\\ 
       in the presence of noiseless feedback\footnote{A preliminary version of this paper 
       was presented as a poster at QIP 2012, 12-16 December 2011, Montr\'{e}al.}}

\author{Runyao Duan}
\email{runyao.duan@uts.edu.au}
\affiliation{Centre for Quantum Computation and Intelligent Systems (QCIS), Faculty of Engineering and Information Technology, University of Technology, Sydney, NSW 2007, Australia}
\affiliation{State Key Laboratory of Intelligent Technology and Systems, Tsinghua National Laboratory for Information Science and Technology, Department of Computer Science and Technology, Tsinghua University, Beijing 100084, China}
\affiliation{UTS-AMSS Joint Research Laboratory for Quantum Computation and Quantum Information Processing, Academy of Mathematics and Systems Science, Chinese Academy of Sciences, Beijing 100190, China}

\author{Simone Severini}
\email{simoseve@gmail.com}
\affiliation{Department of Computer Science and Department of Physics and Astronomy, University College London, WC1E 6BT London, U.K.}

\author{Andreas Winter}
\email{andreas.winter@uab.cat}
\affiliation{ICREA \&{} F\'{\i}sica Te\`{o}rica: Informaci\'{o} i Fen\`{o}mens Qu\`{a}ntics, 
Universitat Aut\`{o}noma de Barcelona, ES-08193 Bellaterra (Barcelona), Spain}
\affiliation{Department of Mathematics, University of Bristol, Bristol BS8 1TW, U.K.}
\affiliation{Centre for Quantum Technologies, National University of Singapore, 2 Science Drive 3, Singapore 117542}

\date{20 April 2016}

\begin{abstract}
We initiate the study of zero-error communication via quantum channels when
the receiver and sender have at their disposal a noiseless feedback channel
of unlimited quantum capacity, generalizing Shannon's zero-error
communication theory with instantaneous feedback.

We first show that this capacity is a function only of the linear span of
Choi-Kraus operators of the channel, which generalizes the bipartite equivocation
graph of a classical channel, and which we dub ``non-commutative bipartite graph''.
Then we go on to show that the feedback-assisted capacity is non-zero 
(allowing for a constant amount of activating noiseless communication)
if and only if the non-commutative bipartite graph is non-trivial,
and give a number of equivalent characterizations. This result involves
a far-reaching extension of the ``conclusive exclusion'' of quantum
states [Pusey/Barrett/Rudolph, \emph{Nature Phys.} {\bf 8}(6):475-478, 2012].

We then present an upper bound on the feedback-assisted zero-error
capacity, motivated by a conjecture originally made by Shannon and 
proved later by Ahlswede. We demonstrate this bound to have
many good properties, including being additive and given by a minimax formula.
We also prove a coding theorem showing that this quantity is the 
entanglement-assisted capacity against an adversarially chosen channel
from the set of all channels with the same Choi-Kraus span, which can also
be interpreted as the feedback-assisted \emph{unambiguous} capacity. 
The proof relies on a generalization of the ``Postselection Lemma''
(de Finetti reduction) [Christandl/K\"onig/Renner, \emph{Phys. Rev. Lett.} {\bf 102}:020504, 2009]
that allows to reflect additional constraints, and which we believe
to be of independent interest.
This capacity is a relaxation of the feedback-assisted zero-error 
capacity; however, we have to leave open the question of whether they 
coincide in general.

We illustrate our ideas with a number of examples, including 
classical-quantum channels and Weyl diagonal channels, 
and close with an extensive discussion of open questions.
\end{abstract}

\maketitle

\thispagestyle{empty}

\vfill\pagebreak

\tableofcontents


\setcounter{page}{1}

\section{Zero-error communication assisted by noiseless quantum feedback}
\label{sec:intro}
In information theory it is customary to consider not only
asymptotically long messages but also asymptotically vanishing,
but nonzero error probabilities, which leads to a probabilistic
theory of communication characterized by entropic capacity
formulas~\cite{Shannon48,CoverThomas}. It is well-known that
when communicating by block codes over a discrete memoryless channel 
at rate below the capacity, the error probability goes to zero 
exponentially in the block length, and while it is one of the major
open problems of information theory to characterize the tradeoff
between rate and error exponent in general, we have by now a fairly 
good understanding of it. 
However, if the error probability is required to vanish faster 
than exponential, or equivalently is required to be zero exactly
(at least in the case of finite alphabets), 
we enter the strange and much less understood realm of zero-error 
information theory~\cite{Shannon56,KoernerOrlitsky}, which 
concerns asymptotic combinatorial problems, most of which
are unsolved and are considered very difficult.
There are a couple of exceptions to this rather depressing
state of affairs, one having been already identified by Shannon in 
his founding paper~\cite{Shannon56}, namely the discrete memoryless 
channel $N(y|x)$ assisted by instantaneous noiseless feedback, whose 
capacity is given by the \emph{fractional packing number} of a 
bipartite graph $\Gamma$ representing the possible transitions $N(y|x)>0$. 
The other one is the the recently considered assistance by no-signalling
correlations~\cite{CLMW:0}, which is also completely solved in terms
the fractional packing number of the same bipartite graph $\Gamma$. 

Recent years have seen attempts to create a theory of quantum
zero-error information theory~\cite{prehist}, identifying some rather strange
phenomena there such as superactivation~\cite{CCH,Duan:zero} or 
entanglement advantage for classical channels~\cite{CLMW-ent-zero,IQC-ent-zero}, 
but resulting also in some general structural progress such as
a quantum channel version of the Lov\'{a}sz number~\cite{DSW:q-theta}.
Motivated by the success in the above-mentioned two models,
two of us in~\cite{DW:ns-ass} (see also~\cite{DuanWang}) have developed 
a theory of zero-error communication over memoryless quantum channels 
assisted by quantum no-signalling correlations, which largely
(if not completely) mirrors the classical channel case; in particular,
it yielded the first capacity interpretation of the Lov\'{a}sz number 
of a graph. Some of the techniques and insights developed in~\cite{DW:ns-ass}
will play a central role also in the present paper.

\medskip
In the present paper, we take as our point of departure the other 
successful case, Shannon's theory of zero-error communication assisted
by noiseless instantaneous feedback. In detail,
consider a quantum channel $\cN:\cL(A) \longrightarrow \cL(B)$, i.e.~a
completely positive and trace preserving (cptp) linear map from the operators
on $A$ to those of $B$ (both finite-dimensional Hilbert spaces), 
where $\cL(A)$ denotes the linear operators (i.e.~matrices) on $A$,
with Choi-Kraus and Stinespring representations
\[
  \cN(\rho) = \sum_j E_j \rho E_j^\dagger = \tr_C V\rho V^\dagger,
\]
for linear operators $E_j : A \longrightarrow B$ such that
$\sum_j E_j^\dagger E_j = \1$, and an isometry $V: A \longrightarrow B\ox C$,
respectively. The linear
span of the Choi-Kraus operators is denoted by 
\[
  K = \cK(\cN) := \operatorname{span} \{ E_j : j \} < \cL(A\rightarrow B),
\]
where ``$<$" means that $K$ is a subspace of $\cL(A\rightarrow B)$, the
linear operators (i.e.~martrices) mapping $A$ to $B$. 
We will discuss a model of communication where Alice uses the channel
$n$ times in succession, allowing Bob after each round to send
her back an arbitrary quantum system. They may also share an entangled
state prior to the first round (if not, they can have it anyway from
the second round on, since Bob could use the first feedback to create
an arbitrary entangled state). Their goal is to allow Alice to
send one of $M$ messages down the channel uses such that Bob
is able to distinguish them perfectly.
More formally, the most general \emph{quantum feedback-assisted code}
consists of a state (w.l.o.g.~pure) $\ket{\phi} \in X_0 \ox Y_0$
and for each message $m=1,\ldots,M$ isometries for encoding and 
feedback decoding
\begin{equation}\begin{split}
  \label{eq:feedback-code}
  U^{(m)}_t &: X_{t-1} \ox F_{t-1} \longrightarrow A_t \ox X_t, \\
  W_t       &: Y_{t-1} \ox B_t     \longrightarrow F_t \ox Y_t,
\end{split}\end{equation}
for $t=1,\ldots,n$ and appropriate local quantum systems $X_t$
(Alice) and $Y_t$ (Bob), as well the feedback-carrying systems
$F_t$; see Fig.~\ref{fig:feedback-diagram}. For consistency 
(and w.l.o.g.), $F_0=F_n=\CC$ are trivial. Note that Bob can use the feedback channel to 
create any entangled state $\ket{\phi}$ with Alice for later use before they actually 
send messages. We use isometries, rather than general cptp maps, to represent 
encoders and decoders in the feedback-assisted communication scheme, because
by the Stinespring dilation~\cite{Stinespring}, all local cptp maps can be 
``purified'' to local isometries. Thus every seemingly more general protocol
involving cptp maps can be purified to one of the above form. We will 
find this form convenient in the later analysis as it allows us to reason
on the level of Hilbert space vectors.

\begin{figure}[ht]
  \includegraphics[width=9cm]{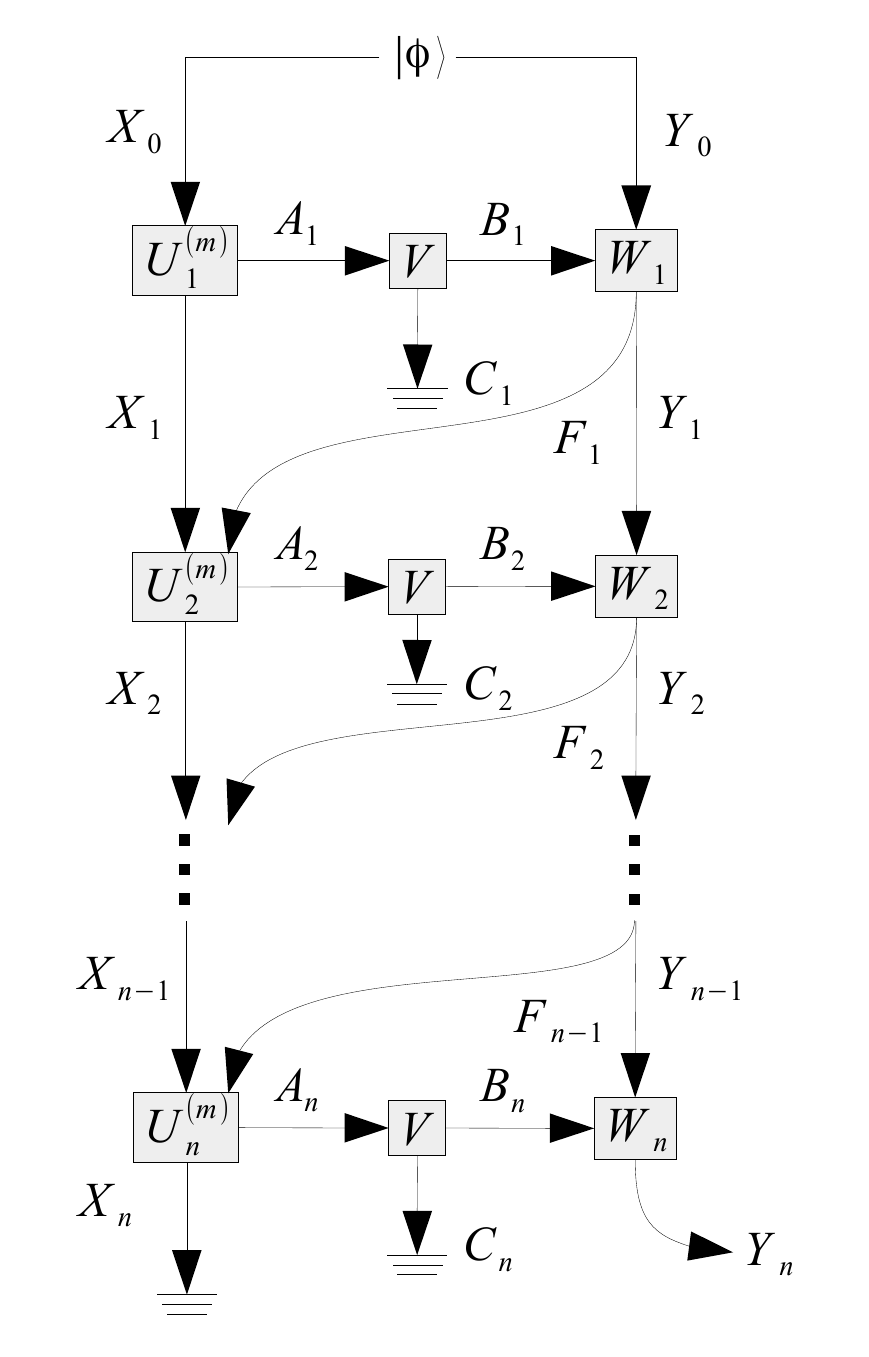}
  \caption{Diagrammatic representation of a feedback-assisted code for
    messages $m$ sent down a channel $\cN$ used $n$ times, in the form
    of a schematic circuit diagram.
    All boxes are isometries (acting on suitably large input and
    output quantum registers), and the solid lines and arrows represent
    the ``sending'' of the respective register.
    Bob's final output state $\rho_m$ after $n$ rounds of using the channel
    and feedback is in register $Y_n$.}
  \label{fig:feedback-diagram}
\end{figure}

We call this quantum feedback-assisted code a \emph{zero-error code} if 
there is a measurement on $Y_n$ that distinguishes Bob's output states
$\rho^{(m)} = \sum_{\underline{j}} \rho^{(m)}_{\underline{j}}$, with 
certainty, where the sum is over the states
\begin{equation}
  \label{eq:feedback-code-outputs}
  \rho^{(m)}_{\underline{j}} 
           = \tr_{X_n} \left( \prod_{t=n}^1 (W_t E_{j_t} U_t^{(m)}) 
                                 \proj{\phi} 
                              \prod_{t=1}^n ({U_t^{(m)}}^\dagger E_{j_t}^\dagger W_t^\dagger) \right),
\end{equation}
which are the output states given a specific sequence $\underline{j}=j_1\ldots j_n$
of Kraus operators. [Note that here and below, for convenience, we use 
$\prod_{t=n}^1 Q_t$ to represent right-to-left multiplications of  operators 
$Q_t$, namely $\prod_{t=n}^1 Q_t:=Q_n\cdots Q_1$.] 
In other words, these states $\rho^{(m)}$ have to have
mutually orthogonal supports, i.e.~for all $m\neq m'$, all $\underline{j}$,
$\underline{k}$ and all $\xi \in \cL(X_n)$,
\[
  0 = \bra{\phi} \prod_{t=1}^n ({U_t^{(m')}}^\dagger E_{j_t}^\dagger W_t^\dagger)
                   \xi
                 \prod_{t=n}^1 (W_t E_{k_t} U_t^{(m)}) \ket{\phi}
    =: \bra{\phi^{(m')}_{\underline{j}}} \xi \ket{\phi^{(m)}_{\underline{k}}}.
\]
By linearity, we see that this condition depends only on the linear
span of the Choi-Kraus operator space $K$, in fact it can evidently be expressed as the orthogonality of a tensor
defined as a function of $\ket{\phi}$, the $U_t^{(m)}$ and $W_t$,
to the subspace $(K \ox K^\dagger)^{\ox n}$ -- cf.~similar albeit
simpler characterizations of zero-error and entanglement-assisted zero-error
codes in terms of the ``non-commutative graph'' 
$S=K^\dagger K:={\rm span}\{E_k^\dag E_j: k,j\} < \cL(A)$~\cite{CCH,Duan:zero,DSW:q-theta}, and of
no-signalling assisted zero-error codes in terms of the
``non-commutative bipartite graph'' $K$~\cite{DW:ns-ass}.
Thus we have proved

\begin{proposition}
  \label{prop:basic}
  A quantum feedback-assisted code for a channel $\cN$ being zero-error 
  is a property solely of the Choi-Kraus space $K = \cK(\cN)$. The maximum
  number of messages in a feedback-assisted zero-error code is denoted
  $M_f(n;K)$.
  Hence, the \emph{quantum feedback-assisted zero-error capacity} of $\cN$,
  \[
    C_{0EF}(K) := \lim_{n\rightarrow\infty} \frac{1}{n} \log M_f(n;K)
                = \sup_n \frac{1}{n} \log M_f(n;K),
  \]
  is a function only of $K$.
  \qed
\end{proposition}

In the case of a classical channel $N:\cX \longrightarrow \cY$
with transition probabilities $N(y|x)$,
assisted by classical noiseless feedback, the above problem was first 
studied -- and completely solved -- by Shannon~\cite{Shannon56}. 
To be precise, his model has noiseless instantaneous feedback
of the channel output back to the encoder; it is clear that any protocol
with general actions (noisy channel acting on the output) 
by the receiver can be simulated by the receiver storing the output
and the encoder getting a copy of the channel output, if shared randomness
is available. Our model differs from this only by the additional 
availability of entanglement; that this does not increase further 
the capacity follows from~\cite{CLMW:0}, see our comments below.

Following Shannon, we introduce the (bipartite) \emph{equivocation graph}
$\Gamma = \Gamma(N)$ on $\cX \times \cY$, which has an edge $xy$ iff 
$N(y|x) > 0$, i.e.~the adjacency matrix is $\Gamma(y|x) = \lceil N(y|x) \rceil$; 
furthermore the confusability graph $G=G(N)$ on $\cX$, with an edge $x \sim x'$ iff 
there exists a $y$ such that $N(y|x)N(y|x')>0$, i.e., iff the neighbourhoods of 
$x$ and $x'$ in $\Gamma$ intersect. The \emph{feedback-assisted zero-error capacity}
$C_{0F}(N)$ of the channel $N$ can be seen to depend only on $\Gamma$.

Note that for (the quantum realisation of) a classical channel, i.e.
\[
  \cN(\rho) = \sum_{xy} N(y|x) \ketbra{y}{x} \rho \ketbra{x}{y},
\]
the corresponding subspace is given by
\[
  K = \operatorname{span}\{ \ketbra{y}{x} : xy \text{ is an edge in } \Gamma \},
\]
so $K$ should really be understood as the quantum generalisation of
the equivocation graph (a \emph{non-commutative bipartite graph})~\cite{DW:ns-ass},
much as $S=K^\dagger K$ was advocated in~\cite{DSW:q-theta} 
as a quantum generalisation of an undirected graph.

Shannon proved
\begin{equation}
  \label{eq:feedback-Shannon}
  C_{0F}(N) = 
  C_{0F}(\Gamma) = \begin{cases}
                     0                     & \text{ if } G \text{ is a complete graph}\ (\text{iff } C_0(N)=0), \\
                     \log \alpha^*(\Gamma) & \text{ otherwise}. 
                   \end{cases}
\end{equation}
Here, $\alpha^*(\Gamma)$ is the so-called \emph{fractional packing number}
of $\Gamma$, defined as a linear programme, whose dual linear programme is the
\emph{fractional covering number}~\cite{Shannon56,ScheinermanUllman}:
\begin{equation}\begin{split}
  \label{eq:fractional-packing-n}
  \alpha^*(\Gamma) &= \max \sum_x w_x \text{ s.t. } \forall x\ 0\leq w_x,\ 
                                                    \forall y\ \sum_x w_x \Gamma(y|x) \leq 1, \\
                   &= \min \sum_y v_y \text{ s.t. } \forall y\ 0\leq v_y,\ 
                                                    \forall x\ \sum_y v_y \Gamma(y|x) \geq 1.
\end{split}\end{equation}
This number appears also in other zero-error communication problems, namely
as the zero-error capacity of the channel assisted by no-signalling
correlations~\cite{CLMW:0}. There, it is also shown to be the asymptotic 
simulation cost of a channel with bipartite graph $\Gamma$ in the presence
of shared randomness. This shows that for a classical channel with
bipartite graph $\Gamma$, interpreted as a quantum channel $\cN$ with 
non-commutative bipartite graph $K$, $C_{0F}(\Gamma) = C_{0EF}(K)$.

The first case in eq.~(\ref{eq:feedback-Shannon})
of a complete graph $G$ is easy to understand: whatever the parties do, 
and regardless of the use of feedback, any two inputs may lead to the same output
sequence, so not a single bit can be transmitted with certainty. In either case,
Shannon showed that only some arbitrarily small rate of perfect communication 
(actually a constant amount, dependent only on $\Gamma$) is sufficient to achieve
what we might call the \emph{activated capacity} $\overline{C}_{0F}(N)$, which
is always equal to $\log\alpha^*(\Gamma)$. This was understood better in the work
of Elias~\cite{Elias:list} who showed that the capacity of zero-error 
\emph{list decoding} of $N$ (with arbitrary but constant list size) is exactly 
$\log \alpha^*(\Gamma)$. Thus a coding scheme for $N$ with feedback would
consist of a zero-error list code with list size $L$ and rate
$R \geq \left(1-\frac1L\right)\log\alpha^*(\Gamma) - O\left(\frac1L\right)$
for $n$ uses of the channel $N$, followed by feedback in which Bob lets Alice 
know the list of $L$ items in which he now knows the message falls, followed 
by a noiseless transmission of $\log L$ bits of Alice to resolve the remaining
ambiguity. Shannon's scheme~\cite{Shannon56} is based on a similar idea,
but whittles down the list by a constant factor in each round, so Bob needs 
to update Alice on the remaining list after each channel use. The constant noiseless
communication at the end of this protocol can be transmitted using an unassisted
zero-error code via the given channel $N$ (at most $\log L$ uses), or
via an activating noiseless channel.

The dichotomy in eq.~(\ref{eq:feedback-Shannon}) has the following 
quantum channel analogue (in fact, generalization):
\begin{proposition}
  \label{prop:zero-cap}
  For any non-commutative bipartite graph $K=\cK(\cN) < \cL(A\rightarrow B)$,
  the feedback-assisted zero-error capacity of $K$ vanishes,
  $C_{0EF}(K) = 0$, if and only if the associated non-commutative
  graph is complete, i.e.~$S=K^\dagger K = \cL(A)$, which is
  equivalent to vanishing entanglement-assisted zero-error capacity,
  $C_{0E}(S) = 0$.
\end{proposition}
\begin{proof}
Clearly $C_{0EF}(K) \geq C_{0E}(S)$ since on the right hand side we simply
do not use feedback, but any code is still a feedback-assisted code.
Hence, if the latter is positive then so is the former. 
It is well known that if $S \neq \cL(A)$, then $C_{0E}(S) \geq 1 > 0$,
in fact each channel use can transmit at least one bit~\cite{Duan:zero,DSW:q-theta}.

Conversely, let us assume that $C_{0E}(S) = 0$, i.e.~$S = K^\dagger K = \cL(A)$. 
We will show by induction on $t$ that for any two distinct messages, w.l.o.g.~$b=0,1$,
Bob's output states after $t$ rounds, $\rho^{(b)}_t$ on $Y_t$, cannot be 
orthogonally supported, meaning $M_f(n;K)=1$.
Here,
\begin{align*}
  \rho^{(b)}_t &= \sum_{j_1\ldots j_t} \tr_{X_t F_t} \proj{\phi^{(b)}_{j_1\ldots j_t}}, \text{ with}\\
  \ket{\phi^{(b)}_{j_1\ldots j_t}}
               &= \prod_{i=t}^1 W_i E_{j_i} U_i^{(b)} \ket{\phi} \in X_t \ox F_t \ox Y_t.
\end{align*}
This is clearly true for $t=0$ since at that point Alice and Bob share only
$\ket{\phi}_{X_0Y_0}$, hence $\rho_0^{(0)} = \rho_0^{(1)} = \tr_{X_0} \phi$.
For $t>0$, let Bob after $t-1$ rounds have one of the states
$\rho^{(b)}_{t-1}$; by the induction hypothesis,
$\rho^{(0)}_{t-1} \not\perp \rho^{(1)}_{t-1}$ -- by a slight
abuse of notation meaning that the supports are not orthogonal, or 
equivalently that the operators are not orthogonal with respect to the 
Hilbert-Schmidt inner product. This means that there are indices
$j_1\ldots j_{t-1}$ and $k_1\ldots k_{t-1}$ such that
\[
  \left(\phi^{(0)}_{t-1}\right)_{Y_{t-1}} := \left(\phi^{(0)}_{j_1\ldots j_{t-1}}\right)_{Y_{t-1}}
     \not\perp 
  \left(\phi^{(1)}_{k_1\ldots k_{t-1}}\right)_{Y_{t-1}} =: \left(\phi^{(1)}_{t-1}\right)_{Y_{t-1}}.
\]
This can be expressed equivalently as 
\begin{equation*}
  \tr_{Y_{t-1}} \ketbra{\phi_{t-1}^{(0)}}{\phi_{t-1}^{(1)}} \neq 0.
\end{equation*}
Now, in the $t$-th round, Alice applies the isometry
$U^{(b)}_t:X_{t-1}F_{t-1} \rightarrow X_t A$ to the $X$ and $F$ registers
of $\ket{\phi^{(b)}_{t-1}}$, hence for $\ket{\psi^{(b)}_t} = U^{(b)}_t\ket{\phi^{(b)}_{t-1}}$
(as we do not touch the $Y_{t-1}$ register)
\begin{equation}
  \label{eq:really-important-step}
  \tr_{Y_{t-1}} \ketbra{\psi^{(0)}_t}{\psi^{(1)}_t} 
     = \tr_{Y_{t-1}} U^{(0)}_t\ketbra{\phi^{(0)}_{t-1}}{\phi^{(1)}_{t-1}}U^{(1)\dagger}_t \neq 0.
\end{equation}
After that, the channel action consists in one of the Choi-Kraus operators
$E_j:A\rightarrow B$. Let us assume, with the aim of establishing a contradiction, that Bob's states
after the channel action were orthogonal, i.e.~for all $j$ and $k$,
\[
  \tr_{X_t} E_j \psi^{(0)}_t E_j^\dagger \perp \tr_{X_t} E_k \psi^{(1)}_t E_k^\dagger.
\]
In other words, for all $j$, $k$ and operators $\xi$ on $X_t$,
\[\begin{split}
  0 &= \bra{\psi^{(1)}_t} \xi \ox E_k^\dagger E_j \ox \1 \ket{\psi^{(0)}_t} \\
    &= \tr\bigl[ (\xi \ox E_k^\dagger E_j) \tr_{Y_{t-1}}\ketbra{\psi^{(0)}_t}{\psi^{(1)}_t} \bigr].
\end{split}\]
But since $\xi$ is arbitrary and the $E_k^\dagger E_j$ span $\mathcal{L}(A)$,
this would imply $\tr_{Y_{t-1}}\ketbra{\psi^{(0)}_t}{\psi^{(1)}_t} = 0$,
contradicting (\ref{eq:really-important-step}).

Thus, applying now also the isometry $W_t:B Y_{y-1} \rightarrow F_t Y_t$, we
find that there exist $j_t$ and $k_t$ such that
\[
  \left(\phi^{(0)}_{j_1\ldots j_{t}}\right)_{F_t Y_{t}}
     \not\perp 
  \left(\phi^{(1)}_{k_1\ldots k_{t}}\right)_{F_t Y_{t}},
    \text{ hence }
  \left(\phi^{(0)}_{j_1\ldots j_{t}}\right)_{Y_{t}}
     \not\perp 
  \left(\phi^{(1)}_{k_1\ldots k_{t}}\right)_{Y_{t}},
\]
and so finally $\rho^{(0)}_t \not\perp \rho^{(1)}_t$, proving the induction step.
\end{proof}

\medskip
Motivated by $\overline{C}_{0F}$ of a classical channel~\cite{Shannon56},
see above, we define also feedback-assisted codes with $n$ channel uses and 
up to $b$ noiseless classical bits of forward communication. The setup
is the same as in eq.~(\ref{eq:feedback-code}) and Fig.~\ref{fig:feedback-diagram}
with $n+b$ rounds, $n$ of which feature the isometric dilation $V$ of $\mathcal{N}$,
and $b$ the isometry $V':\ket{i} \mapsto \ket{i}\ket{i}$ ($i=0,1$) corresponding
to the noiseless bit channel 
$\overline{\id}_2:\rho \mapsto \sum_{i=0}^1 \proj{i} \rho \proj{i}$. It is clear 
that the output states can be written in a way similar to eq.~(\ref{eq:feedback-code-outputs}),
and that the maximum number of messages in a zero-error code depends only
on $n$, $b$ and $K < \mathcal{L}(A\rightarrow B)$, which we denote 
$M_f^{+b}(n;K)$. Clearly, $M_f^{+0}(n;K) = M_f(n;K)$ and in general,
$M_f^{+b+1}(n;K) \geq 2\, M_f^{+b}(n;K)$. Furthermore, it can easily be
verified that
\[
  2^{-b}M_f^{+b}(n;K)\, 2^{-c}M_f^{+c}(m;K) \leq 2^{-b-c}M_f^{+b+c}(n+m;K),
\]
hence we can define the \emph{activated feedback-assisted zero-error capacity}
\[\begin{split}
  \overline{C}_{0EF}(K) := \sup_b \sup_n \frac1n \bigl( \log M_f^{+b}(n;K) - b \bigr) \\
                        = \sup_b \lim_{n\rightarrow\infty} \frac1n \log M_f^{+b}(n;K).
\end{split}\]

Then the above Proposition~\ref{prop:zero-cap} can be rephrased as
\begin{equation}
  \label{eq:feedback-quantum}
  C_{0EF}(K) = \begin{cases}
                 \overline{C}_{0EF}(K) & \text{ if } S = K^\dagger K \neq \cL(A),         \\
                 0                     & \text{ if } S = \cL(A) \ (\text{iff } C_{0E}(S)=0),
               \end{cases}
\end{equation}
motivating our focusing on $\overline{C}_{0EF}(K)$ from now on

\medskip
The rest of the present paper is organized as follows: 
In Section~\ref{sec:feasibility} we start with a concrete example
showing the importance of measurements ``conclusively excluding'' 
hypotheses from a list of options, and go on to show several concise
characterizations of nontrivial channels, i.e.~those for which 
$\overline{C}_{0EF}(K) > 0$.
In Section~\ref{sec:C_min_E} we first review a characterization of
the fractional packing number in terms of the Shannon capacity
minimized over a set of channels, which then 
motivates the definition of $C_{\min E}(K)$ obtained as a minimization of
the entanglement-assisted capacity over quantum channels consistent with 
the given non-commutative bipartite graph. $C_{\min E}(K)$ represents 
the best known upper bound on the feedback-assisted zero-error capacity. 
We illustrate the bound by showing how it allows us to determine
$\overline{C}_{0EF}(K)$ for Weyl diagonal channels, i.e.~$K$ spanned by
discrete Weyl unitaries. We also show that $C_{\min E}(K)$ 
is the ordinary (small error) capacity of the system assisted by entanglement,
against an adversarial choice of the channel
(proof in Appendix~\ref{app:adversarial}, based on a novel Constrained 
Postselection Lemma, aka ``de Finetti reduction'', in Appendix~\ref{app:postselection}).
After that, we conclude in Section~\ref{sec:outro} with a discussion of
open questions and future work.

\section{Characterization of vanishing capacity $\mathbf{\overline{C}_{0EF}(K)}$}
\label{sec:feasibility}
In this section, we will prove the following result.
\begin{theorem}
  \label{thm:C-0EF-feasibility}
  If the non-commutative bipartite graph $K < \mathcal{L}(A\rightarrow B)$
  contains a subspace $\ket{\beta} \ox A^\dag < K$ with a state vector
  $\ket{\beta}\in B$, meaning that the constant channel
  $\mathcal{N}_0:\rho \mapsto \proj{\beta}\tr\rho$ has 
  $\mathcal{K}(\mathcal{N}_0) < K$, then $\overline{C}_{0EF}(K) = 0$;
  we call such $K$ \emph{trivial}.
  
  Conversely, if $K$ is nontrivial, then $\overline{C}_{0EF}(K) > 0$.
\end{theorem}
\begin{proof}{\bf (``trivial $\mathbf{\Rightarrow}$ zero capacity'')}
We show the stronger statement $M_f^{+b}(n;K) = 2^b$ for all $n$ and $b$.
Indeed, as the zero-error condition is only a property of $K$,
we may assume a concrete constant channel $\mathcal{N}_0$
with $\mathcal{K}(\mathcal{N}_0) = \ket{\beta} \ox A^\dag < K$. The outputs
of the $n$ copies of $\mathcal{N}_0$ in the feedback code do not
matter at all as they are going to be $\beta^{\ox n}$, which Bob
can create himself. Hence the only information arriving at Bob's
from Alice is in the $b$ classical bits in the course of the protocol.
But even assisted by entanglement and feedback, Alice can convey
at most $b$ noiseless bits in this way, due to the Quantum Reverse Shannon Theorem \cite{QRST}.
\end{proof}

\medskip
The opposite implication (``nontrivial $\Rightarrow$ positive capacity'')
will be the subject of the remainder of this section. We will start
by looking at cq-channels first -- Subsection~\ref{subsec:pure-cq}
for pure state cq-channels, Subsection~\ref{subsec:mixed-cq-ex} for a
mixed state example and Subsection~\ref{subsec:general-cq} for 
general cq-channels --, before completing the proof for general
channels in Subsection~\ref{subsec:general-channels}.

\subsection{Pure state cq-channels}
\label{subsec:pure-cq}
For a given orthonormal basis $\{\ket{i}\}$ of the input space $A$,
and pure states $\ket{\psi_i}$ in the output space, consider
the cq-channel
\[
  {\cal N}(\rho) = \sum_i \ketbra{\psi_i}{i} \rho \ketbra{i}{\psi_i},
\]
with Kraus subspace
\[
  K := \cK({\cal N}) = \operatorname{span}\{ \ketbra{\psi_i}{i} \}.
\]

We shall demonstrate first the following result:
\begin{proposition}
  For a pure state cq-channel,
  $\overline{C}_{0EF}(K)$ is always positive unless $K$ is trivial,
  which is equivalent to all $\ket{\psi_i}$ being collinear, i.e. 
  $K = \ket{\psi}\otimes A^\dag$ for some pure state $\ket{\psi}$. 
\end{proposition}
\begin{proof}
If $K$ is trivial, then the above proof of the sufficiency of triviality
in Theorem~\ref{thm:C-0EF-feasibility} shows $\overline{C}_{0EF}(K)=0$.

Conversely, if $K$ is non-trivial, then there are two output vectors,
denoted $\ket{\psi_0}$ and $\ket{\psi_1}$, 
that are not collinear, and we shall simply
restrict the channel to the corresponding inputs $0$ and $1$. 
I.e., we focus only on
$K' = \operatorname{span}\{ \ketbra{\psi_0}{0}, \ketbra{\psi_1}{1} \}$ ,
and the corresponding channel
\[
  {\cal N}'(\rho) = \ketbra{\psi_0}{0} \rho \ketbra{0}{\psi_0}
                    + \ketbra{\psi_1}{1} \rho \ketbra{1}{\psi_1}.
\]
Consider using it three times, inputting only the code words
$001$, $010$ and $100$. This gives rise to output states
\begin{align*}
  \ket{u_a} &= \ket{\psi_0}\ket{\psi_0}\ket{\psi_1}, \\
  \ket{u_b} &= \ket{\psi_0}\ket{\psi_1}\ket{\psi_0}, \\
  \ket{u_c} &= \ket{\psi_1}\ket{\psi_0}\ket{\psi_0},
\end{align*}
which have the property that their pairwise inner products are all equal:
$\langle u_x \ket{u_y} = |\langle \psi_0 \ket{\psi_1} |^2 =: \epsilon$.
By using the channel $3n$ times, Alice can prepare the states
\[
  \ket{t_x} = \ket{u_x}^{\otimes n} \quad (x=a,b,c),
\]
whose pairwise inner products are all equal and indeed $\epsilon^n$,
i.e.~arbitrarily close to $0$. Now, if $n$ is large enough (so that
$\epsilon^n \leq \frac{1}{2}$), there is a cptp map that Bob can apply to
transform
\begin{align*}
  \ket{t_a} &\longmapsto \frac{1}{\sqrt{2}}(\ket{1}+\ket{2}), \\
  \ket{t_b} &\longmapsto \frac{1}{\sqrt{2}}(\ket{2}+\ket{0}), \\
  \ket{t_c} &\longmapsto \frac{1}{\sqrt{2}}(\ket{0}+\ket{1}),
\end{align*}
(This follows from well known results on pure-state transformations,
see e.g.~\cite{CJW-pure-trans}.)
By now it may be clear where this is going: Bob measures the computational
basis and overall we obtain a classical channel
$P: \{a,b,c\} \rightarrow \{0,1,2\}$
with exactly one $0$-entry in each row and column:
\[
  P(0|a) = P(1|b) = P(2|c) = 0,
\]
which has zero-error capacity $0$, but assisted by feedback and a
finite number of activating noiseless bits, 
it is $\log\frac{3}{2}$~\cite{Shannon56}.
We conclude that $\overline{C}_{0EF}({\cal N}) \geq \frac{1}{3n}\log\frac{3}{2} > 0$.
\end{proof}

\subsection{Mixed state cq-channel}
\label{subsec:mixed-cq-ex}
To generalize the previous treatment to mixed states, let us
first look at a specific simple example:
Let $\ket{\psi_i}$ ($i=0,1,2$) be three mutually distinct 
but non-orthogonal states in $\CC^3$,
and define a cq-channel $\cN$ with three inputs $i=0,1,2$, mapping
\begin{align}
  0 &\longmapsto \frac{1}{2}\psi_1 + \frac{1}{2}\psi_2, \nonumber\\
  \label{eq:mixed-cq}
  1 &\longmapsto \frac{1}{2}\psi_0 + \frac{1}{2}\psi_2, \\
  2 &\longmapsto \frac{1}{2}\psi_0 + \frac{1}{2}\psi_1. \nonumber
\end{align}
Thus,
\[
  K = \operatorname{span}\{ \ketbra{\psi_1}{0},\ketbra{\psi_2}{0},
                            \ketbra{\psi_0}{1},\ketbra{\psi_2}{1},
                            \ketbra{\psi_0}{2},\ketbra{\psi_1}{2} \},
\]
and the most general channel $\cN'$ consistent with this $K$ is a 
cq-channel of the form
\[
  i \longmapsto \rho_i, \quad \rho_i \text{ supported on }
                        \operatorname{span}\bigl\{ \ket{\psi_j}: j \in \{0,1,2\}\setminus i \bigr\}.
\]

We shall show how to construct a zero-error scheme with feedback,
achieving positive rate, at least for $\ket{\psi_i}$ that are sufficiently
close to being orthogonal.
For the zero-error properties, we may as well focus on $\cN$,
which is easier to reason with. For the following, it may be helpful
to think of eq.~(\ref{eq:mixed-cq}) in a partly classical way: any
input $i$ is mapped to a random $\ket{\psi_j}$, subject to $j\neq i$,
so that for two uses of the channel, each pair $i_1i_2$ is
mapped randomly to one of four $\ket{\psi_{j_1}}\ket{\psi_{j_2}}$,
with $j_1\neq i_1$, $j_2\neq i_2$. Of course, vice versa each of these
nine vectors is reached from exactly four inputs.

Now, assuming that the pairwise inner products of the 
$\ket{\psi_i}$ are small enough, i.e.
\[
  |\bra{\psi_0} \psi_1 \rangle|, \ 
  |\bra{\psi_0} \psi_2 \rangle|, \ 
  |\bra{\psi_1} \psi_2 \rangle| \leq \epsilon,
\]
to guarantee that there is a deterministic pure state transformation (by cptp map)
$\ket{\psi_{j_1}}\ket{\psi_{j_2}} \longmapsto \ket{\varphi_{j_1j_2}}$~\cite{CJW-pure-trans},
where
\[
  \ket{\varphi_{j_1j_2}} = \frac{1}{\sqrt{8}} \sum_{j_1j_2 \in I \subset \{0,1,2\}^2 \atop |I|=2} \ket{I}
                         \in \CC^{36}.
\]
On these states, Bob performs a measurement in the computational basis
of the $\ket{I}$, and we get an effective classical channel
mapping $i_1i_2 \in \{0,1,2\}^2$ randomly to some
$\{j_1j_2, k_1k_2\} = I \subset \{0,1,2\}^2$, subject to the constraint
\[
  (j_1\neq i_1 \text{ \& } j_2\neq i_2) \text{ or }
  (k_1\neq i_1 \text{ \& } k_2\neq i_2),
\]
which means that each $I$ is reached from at most eight out of the nine
pairs $i_1i_2$.
In fact, the observation of $I = \{ j_1j_2,k_1k_2\}$ excludes at least 
two out of nine input symbols, namely $j_1k_2$ and $k_1j_2$, 
meaning that this classical channel has zero-error capacity (plus feedback 
plus a finite number of noiseless bits) of 
$\geq \log \frac{9}{7}$. In conclusion, we achieve for $\cN$, and hence for 
any $\cN'$ with $\cK(\cN') < K$, a rate of $\geq \frac{1}{2}\log\frac{9}{7} > 0$.
\qed

\subsection{General cq-channels}
\label{subsec:general-cq}
The above examples rely on measuring the output states $\rho_i$ of the
cq-channel $\mathcal{N}$ by a POVM $(M_j)$ such that the resulting 
classical(!) channel $N:i \rightarrow j$ with $N(j|i) = \tr \rho_i M_j$
has an equivocation graph $\Gamma$ with $\alpha^*(\Gamma) > 1$, because then 
$\overline{C}_{0EF}(K) \geq \overline{C}_{0F}(\Gamma) = \log \alpha^*(\Gamma) > 0$.
For this, cf.~eq.~(\ref{eq:fractional-packing-n}), it is necessary and
sufficient that each outcome $j$ excludes at least one input $i$,
i.e.~$N(j|i) = \tr \rho_i M_j = 0$, or equivalently $\rho_i \perp M_j$.
A POVM $(M_j)$ with this property is said to ``conclusively exclude''
the set $\{\rho_i\}$ of states~\cite{PBR,BJOP}.
It is clearly only a property of the support projections $P_i$ of
$\rho_i$, and w.l.o.g.~the POVM is indexed by the same $i$'s, i.e.~$(R_i)$
such that $P_i R_i = 0$ for all $i$, as well as $R_i \geq 0$ and
$\sum_i R_i = \1$. 

Our approach in the following will be to characterize
when a set $\{\rho_i\}$ of states, or one of its tensor powers
$\{\rho_i\}^{\ox n} = \{ \rho_{\underline{i}} = \rho_{i_1}\ox\cdots\ox\rho_{i_n} \}$,
can be conclusively excluded.
For instance, Pusey, Barrett and Rudolph~\cite{PBR} showed that for
any two linearly independent pure states $\ket{\psi_0}$ and $\ket{\psi_1}$, 
it is always possible to find an integer $n$ and a $2^n$-outcome POVM 
$\bigl( R_{\underline{i}}: \underline{i}\in \{0,1\}^n \bigr)$ such that 
\[
  \tr R_{\underline{i}} \proj{\psi_{\underline{i}}} = 0,
  \quad 
  \ket{\psi_{\underline{i}}} = \ket{\psi_{i_1}}\ox \ket{\psi_{i_2}}\ox \cdots \ox \ket{\psi_{i_n}}.
\]
I.e.~we can design a quantum measurement that can conclusively exclude the 
$n$-fold states $\ket{\psi_{\underline{i}}}$ with $n$-bit strings 
$\underline{i}=i_1\ldots i_n$ as outcomes, even when $\ket{\psi_0}$ and 
$\ket{\psi_1}$ are not orthogonal. 

We will employ the powerful techniques developed in the proof
of~\cite[Prop.~14]{DW:ns-ass}, 
allowing us to show a far-reaching generalization of the 
Pusey/Barrett/Rudolph result~\cite{PBR}. The version we need
can be stated as follows; it is adapted to a cq-channel with
$a$-dimensional input space $A$ and output states $\rho_i$ ($i=1,\ldots,a$),
whose support projectors are $P_i$ and supports $K_i$, so that
the non-commutative graph is 
\[
K = \sum_i K_i \ox \bra{i}:={\rm span}\{\ketbra{\psi_i}{i}: \ket{\psi_i}\in K_i, i=1,\cdots, a\}.
\]

\begin{proposition}
  \label{key-lemma-2}
  Let $(P_i)_{i=1}^a$ be projectors on a Hilbert space $B$,
  with a transitive group action by unitary conjugation on the $P_i$,
  i.e.~we have a finite group $G$ acting transitively on the 
  labels $i$, and a unitary representation $U^g$ such that 
  $P_{i^g} = (U^g)^\dagger P_i U^g$ for $g \in G$.
  
  Consider the isotypical decomposition of $U^g$,
  \[
    B = \bigoplus_\lambda \cQ_\lambda \ox \cR_\lambda
  \]
  into irreps $\cQ_\lambda$ of $U^g$, with multiplicity spaces $\cR_\lambda$
  (cf.~\cite{FH1991}, see also~\cite{Harrow2005,Christandl2006}).
  Denote the number of
  terms $\lambda$ by $L$, and the largest occurring multiplicity by
  $M = \max_\lambda |\cR_\lambda|$. If now
  \[
    \frac{a}{\left\|\sum_i P_i\right\|_\infty} > 16 L^6 M^9,
  \]
  then there exists a POVM $(R_i)$ with $P_i R_i = 0$ for all $i$.
  In other words, any set $\{\rho_i\}$ with $\operatorname{supp} \rho_i < K_i$
  can be conclusively excluded.
\end{proposition}

\medskip
Before we prove it, we use it to derive the following general result. To state
it, we need some notation: For a set $\cE=\{\rho_i\}_{i=1}^a$ of states, let
\[
  \cE^{\ox n} = \bigl\{ \rho_{\underline{i}} = \rho_{i_1}\ox\cdots\ox\rho_{i_n}
                             : \underline{i} = i_1\ldots i_n \in [a]^n \bigr\}.
\]
The strings $\underline{i} = i_1\ldots i_n$ are classified according
to \emph{type} $\tau$~\cite{Csiszar:method-of-types}, 
which is the empirical distribution of the letters
$i_t$, $t=1,\ldots,n$. There are only ${n+a-1 \choose a-1} \leq (n+1)^a$
many different types. The subset of $\cE^{\ox n}$ corresponding to type $\tau$
is denoted
\[
  \cE^{(n)}_\tau = \bigl\{ \rho_{\underline{i}} = \rho_{i_1}\ox\cdots\ox\rho_{i_n}
                             : \underline{i} = i_1\ldots i_n \text{ has type } \tau \bigr\}.
\]

We also recall the definition of the \emph{semidefinite packing number}~\cite{DW:ns-ass}
of a non-commutative bipartite graph $K$ with support projection $P_{AB}$
onto the Choi-Jamio\l{}kowski range $(\1 \ox K)\ket{\Phi}$, where
$\ket{\Phi} = \frac{1}{\sqrt{|A|}}\sum_{i=1}^{|A|} \ket{i}\ket{i}$
is the maximally entangled state:
\begin{equation}\begin{split}
  \label{eq:Aram}
  \Aram(K) &= \max \tr S_A \ \text{ s.t. }\  0 \leq S_A,\ \tr_A P_{AB}(S_A\ox\1_B) \leq \1_B \\
           &= \min \tr T_B \ \text{ s.t. }\  0 \leq T_B,\ \tr_B P_{AB}(\1_A\ox T_B) \geq \1_A.
\end{split}\end{equation}
For the cq-channel case, $P_{AB}=\sum_i \proj{i}^A \ox P_i^B$, this simplifies to
\begin{equation}
  \label{eq:Aram-cq}
  \Aram(K) := \max \sum_i s_i \ \text{ s.t. }\ 0\leq s_i,\ \sum_i s_i P_i\leq \1.
\end{equation}
In particular,   for the cq-graph $K$ induced by projections $\{P_i\}$ in Proposition \ref{key-lemma-2}, we have
\[
A(K)= \frac{a}{\left\|\sum_i P_i\right\|_\infty}.
\]

\begin{theorem}
  \label{exclude-discrimination}
  Let $\cE=\{\rho_i\}_{i=1}^a$ be a finite set of quantum states with supports
  $K_i = \operatorname{supp} \rho_i$, and let $K$ 
  be the associated non-commutative bipartite graph
  $\sum_i K_i \otimes \bra{i}$. 
  Then the following are equivalent:
  \renewcommand{\theenumi}{\roman{enumi}}
  \begin{enumerate}
    \item $\overline{C}_{0EF}(K) > 0$;
    \item $K$ is nontrivial;
    \item $\bigcap_i K_i = 0$;
    \item $\left\|\sum_i P_i\right\|_\infty < a$;
    \item $\Aram(K)>1$;
    \item For sufficiently large $n$ and a suitable type $\tau$, the
      set $\cE^{(n)}_\tau$ can be conclusively excluded.
  \end{enumerate}
  \renewcommand{\theenumi}{\arabic{enumi}}
\end{theorem}
\begin{proof}
\emph{i. $\Rightarrow$ ii.} has been shown in the first part (necessity)
  of Theorem~\ref{thm:C-0EF-feasibility}, at the start of this section.

\emph{ii. $\Leftrightarrow$ iii.} $\ket{\beta}\ox A^\dag < K = \sum_i K_i \ox \bra{i}$
  if and only if $\ket{\beta} \in \bigcap_i K_i$.

\emph{iii. $\Leftrightarrow$ iv.} 
  $\left\| \sum_i P_i \right\|_\infty \leq \sum_i \| P_i \|_\infty = a$ with
  equality if and only if there is a common eigenvector $\ket{\beta}$
  with eigenvalue $1$ for all of the $P_i$, i.e.~$\ket{\beta} \in \bigcap_i K_i$.

\emph{iv. $\Rightarrow$ v.} We check that 
  $s_i = \frac{1}{\left\| \sum_i P_i \right\|_\infty}$ is feasible for $\Aram(K)$;
  indeed,
  \[
    \sum_i s_i P_i = \frac{1}{\left\| \sum_i P_i \right\|_\infty} \sum_i P_i \leq \1,
  \]
  thus $\Aram(K) \geq \frac{a}{\left\| \sum_i P_i \right\|_\infty} > 1$.

\emph{v. $\Rightarrow$ vi.} Note that the non-commutative bipartite graph
  corresponding to $\cE^{\ox n}$ is $K^{\ox n}$. Let's denote the graph
  of $\cE^{(n)}_\tau$ by $K^{(n)}_\tau$.
  In~\cite{DW:ns-ass} it is shown that
  $\Aram(K)$ is multiplicative, $\Aram(K^{\ox n}) = \Aram(K)^n$; indeed,
  for an optimal assignment of weights $s_i$ feasible for $\Aram(K)$,
  $s_{\underline{i}} = s_{i_1}\cdots s_{i_n}$ is feasible (and optimal)
  for $\Aram(K^{\ox n})$. 
  Hence, there exists a type $\tau$ such that 
  \begin{equation}
    \label{eq:Aram-exponential}
    \Aram(K^{(n)}_\tau) \geq \sum_{\underline{i}\in\tau} s_{\underline{i}} 
                        \geq \frac{1}{\text{poly}(n)} \Aram(K)^n.
  \end{equation}
  On the other hand, the symmetric group $S_n$ acts transitively by permutation
  on the strings of type $\tau$, and equivalently by permutation of the
  $n$ tensor factors of $B^n$. This representation is well known to have
  only $L \leq \text{poly}(n)$ irreps, each of which has multiplicity
  $M \leq \text{poly}(n)$. Thus, from eq.~(\ref{eq:Aram-exponential}), we
  deduce that for sufficiently large $n$, $\Aram(K^{(n)}_\tau) \geq 16L^6M^9$,
  which by Proposition~\ref{key-lemma-2} implies that the set
  $\cE^{(n)}_\tau$ can be conclusively excluded.

\emph{vi. $\Rightarrow$ i.} By sending signals 
  $\underline{i}=i_1\ldots i_n \in \tau$ and measuring the output states 
  $\rho_{\underline{i}}$ with a conclusively
  excluding POVM $(M_{\underline{i}}:\underline{i}\in\tau)$, we simulate
  a classical channel whose bipartite equivocation graph $\Gamma$ has 
  $\alpha^*(\Gamma) > 1$,
  hence $\overline{C}_{0EF}(K) \geq \frac1n \overline{C}_{0F}(\Gamma) > 0$.
\end{proof}

\medskip
\begin{proof}{\bf (of Proposition~\ref{key-lemma-2})}
Assume that we have a feasible $s_i=s^*$ ($i=1,\ldots,a$)
for $\Aram(K)$ such that $A(K)\geq \sum_i s_i=s^*a \geq 16 L^6 M^{9}$. 
Concretely, this means that $\sum_i s_i P_i = s^* \sum_i P_i \leq \1$.

We will show that a desired POVM $(R_i)$ 
can be found, such that $R_{i^g} = (U^{g})^\dag R_i U^g$ for all $i$
and $g$. 
The problem of finding the POVM $(R_i)$ then becomes equivalent to 
finding $0\leq R_0 \leq \1-P_0$ such that
\begin{equation}
  \label{eq:schur-idea}
  \frac{1}{a}\sum_{i=1}^a R_i = \frac{1}{|G|}\sum_{g\in G} (U^g)^\dag R_0 U^g = \frac{1}{a}\1.
\end{equation}
Schur's Lemma~\cite{FH1991} tells us
\[
  \frac{1}{|G|} \sum_g (U^g)^\dagger R_0 U^g 
         = \frac{1}{a} \sum_\lambda Q_\lambda \ox \zeta_\lambda,
\]
where $Q_\lambda$ is the projection onto the irrep $\cQ_\lambda$,
$\zeta_\lambda$ is a semidefinite operator on $\cR_\lambda$.
The equality constraints (\ref{eq:schur-idea}) on $R_0$ are equivalent 
to $\zeta_\lambda = \Pi_\lambda$, the projection onto $\cR_\lambda$,
for all $\lambda$. 

Now, for each $\lambda$ choose an orthogonal basis $\{Z^{(\lambda)}_\mu\}$
of Hermitians over $\cR_\lambda$, 
with $Z^{(\lambda)}_0 = \frac{1}{\tr\Pi_\lambda}\Pi_\lambda$
and $\| Z^{(\lambda)}_\mu \|_2 = 1$ for $\mu \neq 0$. Then the operators
$\frac{1}{\tr Q_\lambda}Q_\lambda \ox Z^{(\lambda)}_\mu$ form a basis of
the $U^g$-invariant operators, hence our constraints on $R_0$ can be rephrased as
\begin{equation}
  \label{eq:feasibility-conds}
  0 \leq R_0 \leq \1-P_0,\quad 
  \tr R_0\!\left (\frac{Q_\lambda}{\tr Q_\lambda}\ox Z^{(\lambda)}_\mu \right)
                                                   = \frac{1}{a}\delta_{\mu 0}\ \forall~\lambda, \mu.
\end{equation}
Notice that here, the semidefinite constraints on $R_0$ leave quite some
room, whereas we have ``only'' $LM^2$ linear conditions to satisfy.
Given $s^*$ satisfying the constraint of $\Aram(K)$, our strategy now will be
to show that we can construct a $0 \leq R_0 \leq \frac{2}{a}(\1-P_0)$ such that 
Eqs. (\ref{eq:feasibility-conds}) hold.

In detail, introduce a new variable $X\geq 0$, with
\[
  R_0 = \frac{1}{a}(\1-P_0)X(\1-P_0),
\]
which makes sure that $R_0$ is automatically supported on the complement of $P_0$.
Now rewrite the conditions (\ref{eq:feasibility-conds})
in terms of $X$, introducing the notation
\[
  C_{\lambda\mu} = \frac{1}{\tr Q_\lambda}Q_\lambda \ox Z^{(\lambda)}_\mu,
  \quad
  D_{\lambda\mu} = (\1-P_0) C_{\lambda\mu} (\1-P_0).
\]
This gives the new form of the constraints as
\begin{equation} 
  \label{eq:X}
  \tr X D_{\lambda\mu} = \delta_{\mu 0}.
\end{equation}

Our goal will be to find a ``nice'' dual set $\{\widehat{D}_{\lambda\mu}\}$ to the 
$\{D_{\lambda\mu}\}$, 
i.e.~$\tr D_{\lambda\mu} \widehat{D}_{\lambda'\mu'} = \delta_{\lambda\lambda'}\delta_{\mu\mu'}$,
with which we can write a solution
$X = \sum_{\lambda\mu} \delta_{\mu 0} \widehat{D}_{\lambda\mu}
   = \sum_\lambda \widehat{D}_{\lambda 0}$.
To this end, we construct first the dual set $\widehat{C}_{\lambda\mu}$ of the 
$\{C_{\lambda_\mu}\}$, which is easy:
\[
  \widehat{C}_{\lambda\mu} = Q_\lambda \ox \widehat{Z}^{(\lambda)}_\mu
                           = \begin{cases}
                               Q_\lambda \ox \Pi_\lambda       & \text{ for } \mu = 0, \\
                               Q_\lambda \ox Z^{(\lambda)}_\mu & \text{ for } \mu\neq 0,
                             \end{cases}
\]
so that indeed 
$\tr C_{\lambda\mu} \widehat{C}_{\lambda'\mu'} = \delta_{\lambda\lambda'}\delta_{\mu\mu'}$.
Now, consider the $LM^2\times LM^2$-matrix $T$,
\[\begin{split}
  T_{\lambda\mu,\lambda'\mu'} &= \tr D_{\lambda\mu} \widehat{C}_{\lambda'\mu'}             \\
                &= \tr (\1-P_0)C_{\lambda\mu}(\1-P_0) \widehat{C}_{\lambda'\mu'}            \\
                &= \delta_{\lambda\lambda'}\delta_{\mu\mu'} - \Delta_{\lambda\mu,\lambda'\mu'},
\end{split}\]
with the deviation
\[
  \Delta_{\lambda\mu,\lambda'\mu'} = \tr P_0 C_{\lambda\mu} (\1-P_0) \widehat{C}_{\lambda'\mu'}
                                     + \tr C_{\lambda\mu} P_0 \widehat{C}_{\lambda'\mu'}.
\]
Here, 
\[\begin{split}
  |\Delta_{\lambda\mu,\lambda'\mu'}| 
            &\leq 2 \| P_0 C_{\lambda\mu} \|_1 \|\widehat{C}_{\lambda'\mu'}\|_\infty \\
            &\leq 2 \| P_0 C_{\lambda\mu} \|_1 
             =    2 \bigl\| P_0 |C_{\lambda\mu}| \bigr\|_1                           \\
            &\leq 2\sqrt{\tr P_0 |C_{\lambda\mu}|}\sqrt{\| C_{\lambda\mu}\|_1},
\end{split}\]
using $\|\widehat{C}_{\lambda'\mu'}\|_\infty \leq 1$, the unitary invariance of
the trace norm, and Lemma~\ref{lemma:tracenorm-trace} stated below. Since
$|C_{\lambda\mu}| = \frac{1}{\tr Q_\lambda} Q_\lambda \ox |Z^{(\lambda)}_\mu|$
is invariant under the action of $U^g$, we have 
$\tr P_0 |C_{\lambda\mu}| = \tr P_i |C_{\lambda\mu}|$ for all $i$, and using
$\sum_i s^* P_i \leq \1$ we get
\begin{equation}
  \label{eq:Delta-bound-2}
  |\Delta_{\lambda\mu,\lambda'\mu'}| \leq 2\sqrt{\frac{1}{s^*a}\|C_{\lambda\mu}\|_1^2}
                                     \leq 2\sqrt{M} (s^*a)^{-1/2}.
\end{equation}
With this and introducing a new parameter $\beta$ we get that
\begin{equation}\begin{split}\label{T-bound-2}
  \| T-\1 \|_\infty \leq \| T-\1 \|_2 
                    &=    \sqrt{ \sum_{\lambda\mu\lambda'\mu'} |\Delta_{\lambda\mu,\lambda'\mu'}|^2} \\
                    &\leq \sqrt{ L^2M^4 4 M (s^*a)^{-1}}
                     \leq \frac{1}{\beta},
\end{split}\end{equation}
where $s^*a \geq 4\beta^2 L^2 M^{5}$. Assuming $\beta \geq 2$ (which will be the
case with our later choice), we thus know that $T$ is invertible; in fact,
we have $T = \1 - \Delta$ with $\| \Delta \|_\infty \leq \frac{1}{\beta} \leq \frac12$, hence
$T^{-1} = \sum_{k=0}^\infty \Delta^k$ and so
\[
  \left\| T^{-1} - \1 \right\|_{\infty} 
                           =    \left\| \sum_{k=1}^\infty \Delta^k \right\|_{\infty}
                           \leq \sum_{k=1}^\infty \| \Delta \|_{\infty}^k
                           \leq    \frac{1}{\beta-1}
                           \leq \frac{2}{\beta}.
\]
I.e., writing $T^{-1} = \1 + \widetilde{\Delta}_{\lambda\mu,\lambda'\mu'}$ we get
\begin{equation}
  \label{eq:tilde-Delta-bound-2}
  |\widetilde{\Delta}_{\lambda\mu,\lambda'\mu'}| \leq \| \widetilde{\Delta} \|_\infty \leq \frac{2}{\beta}.
\end{equation}
The invertibility of $T$ implies that there is a dual set to $\{D_{\lambda\mu}\}$ in 
$\operatorname{span}\{\widehat{C}_{\lambda\mu}\}$. Indeed, from the definition of 
$T_{\lambda\mu,\lambda'\mu'}$ and the dual sets,
\begin{align*}
  \widehat{C}_{\lambda'\mu'} 
           &= \sum_{\lambda\mu} T_{\lambda\mu,\lambda'\mu'} \widehat{D}_{\lambda\mu},
              \text{ which can be rewritten as} \\
  \widehat{D}_{\lambda\mu} 
           &= \sum_{\lambda'\mu'} (T^{-1})_{\lambda'\mu',\lambda\mu} \widehat{C}_{\lambda'\mu'}.
\end{align*}

Now we can finally write down our candidate solution to Eq.~(\ref{eq:X}):
\[\begin{split}
  X &= \sum_{\lambda\mu} \delta_{\mu 0} \widehat{D}_{\lambda\mu} \\
    &= \sum_{\lambda\mu} \delta_{\mu 0} 
               \sum_{\lambda'\mu'} (T^{-1})_{\lambda'\mu',\lambda\mu} \widehat{C}_{\lambda'\mu'} \\
    &= \sum_{\lambda} \widehat{C}_{\lambda 0}                                                    
        +\sum_{\lambda\lambda'\mu'} \widetilde{\Delta}_{\lambda'\mu',\lambda 0} 
                                                      \widehat{C}_{\lambda'\mu'} \\
    &= \1 + \text{Rest}.
\end{split}\]
The rest term can be bounded as follows:
\[\begin{split}
  \| \text{Rest} \|_\infty &\leq \sum_{\lambda\lambda'\mu'} \frac{2}{\beta}
                    =\frac{2}{\beta} L^2M^2 
                           \end{split}\]
using Eq.~(\ref{eq:tilde-Delta-bound-2}). Thus we find
$\| \text{Rest} \|_\infty \leq 1$ if $\beta \geq 2L^2 M^2$ and
$s^*a \geq 4\beta^2 L^2 M^{5} \geq 16 L^6 M^{9}$. 
In this case, we will have $0\leq X \leq 2$ and 
$R_0 := \frac{1}{a}(\1-P_0)X(\1-P_0)$ satisfies 
\[
  0 \leq R_0 \leq \frac{2}{a}(\1-P_0) \leq \1-P_0,
\]
as well as 
\[
  \frac{1}{|G|} \sum_{g\in G} (U^g)^\dag  R_0 U^g = \frac{1}{a}\1.
\]
Thus we get the desired POVM 
$\left( R_i = \frac{a}{|G|}\sum_{g \text{ s.t. } 0^g=i}(U^{g})^\dag R_0 U^g \right)$ 
such that  
\[
  \sum_i R_i = \1,\ R_i\geq 0,\ \tr P_iR_i = 0,
\]
and we are done.
\end{proof}

\begin{lemma}[Lemma 15 in~\cite{DW:ns-ass}]
\label{lemma:tracenorm-trace}
Let $\rho$ be a state and $P$ a projection in a Hilbert space $\cH$. Then,
\[
  \tr \rho P \leq \| \rho P \|_1 \leq \sqrt{\tr\rho P}.
\]
More generally, for $X\geq 0$ and a POVM element $0\leq E \leq \1$,
\[
  \phantom{===================:}
  \tr XE \leq \| XE \|_1 \leq \sqrt{\tr X}\sqrt{\tr XE}.
  \phantom{==============}\qed
\]
\end{lemma}

%

\bigskip
We even recover the Pusey/Barrett/Rudolph result~\cite{PBR} as a corollary:
There, $\cE=\{\ket{\psi_0},\ket{\psi_1}\}$ with (w.l.o.g.)
$\ket{\psi_{0,1}} = \alpha\ket{0}\pm\beta\ket{1}$ qubit states, 
$1 > \alpha \geq \beta > 0$. We have the unitary phase action 
of $\ZZ_2 = \{\1,Z\}$, $Z\ket{\psi_{0,1}} = \ket{\psi_{1,0}}$, and hence
on $\cE^{\ox n}$ we have a transitive action of $G = \ZZ_2^n \rtimes S_n$
(the semidirect product), the symmetric group $S_n$ acting by
permutation of the tensor factors and $\ZZ_2^n$ as $\bigotimes_{t=1}^n Z^{b_t}$.
It has $L=n$ and $M\leq n+1$~\cite{FH1991}, whereas 
\[
  \frac{a}{\left\| \sum_{\underline{i}} P_{\underline{i}} \right\|_\infty}
                        = \frac{2^n}{(2\alpha^2)^n} = \frac{1}{\alpha^{2n}}. 
\]
Hence, for large enough $n$, we have that the latter exceeds $16 L^6 M^9 = \text{poly}(n)$,
and then Proposition~\ref{key-lemma-2} above implies that $\cE^{\ox n}$ can
be conclusively excluded.
\qed

\subsection{General case}
\label{subsec:general-channels}
We shall reduce the case of a general channel to that of a cq-channel.
Indeed, recall that we allow Alice and Bob to share entanglement,
so Alice can encode information into the Bell states
\[
  \ket{\Phi_{uv}} = (\1 \ox Z^v X^u)\ket{\Phi} = (X^u Z^v \ox \1)\ket{\Phi},
\]
with the maximally entangled state 
$\ket{\Phi} = \frac{1}{\sqrt{|A|}}\sum_{i=1}^{|A|} \ket{i}\ket{i}$
and the discrete Weyl operators $X$ and $Z$ (basis and phase shift).
This effectively constructs a cq-channel (with $a=|A|$)
\begin{equation}
  \label{eq:M-channel}
  \cM:[a]^2 \ni uv \mapsto (\id\ox\cN)\proj{\Phi_{uv}} 
                           = (X^u Z^v \ox \1)\rho_{00}(Z^{-v} X^{-u} \ox \1) \in \cS(AB),
\end{equation}
with the Choi-Jamio\l{}kowski state $\rho_{00} = (\id\ox\cN)\proj{\Phi}$.
Applying Theorem~\ref{exclude-discrimination} to this channel is the key to
obtain the following result, which in turn directly implies the
reverse direction (``nontrivial $\Rightarrow$ positive capacity'')
in Theorem~\ref{thm:C-0EF-feasibility}, concluding its proof.

\begin{proposition}
  \label{prop:feasibility-equiv} 
  A non-commutative bipartite graph $K$ with support projection $P_{AB}$
  onto the Choi-Jamio\l{}kowski range $(\1 \ox K)\ket{\Phi}$ has positive activated 
  feedback assisted zero-error capacity, $\overline{C}_{0EF}(K) > 0$, 
  if and only if one of the following equivalent conditions hold:
  \renewcommand{\theenumi}{\roman{enumi}}
  \begin{enumerate}
    \item $K$ is non-trivial, i.e.~there is no constant channel $\mathcal{N}_0$
      with $\mathcal{K}(\mathcal{N}_0) < K$;
    \item There is no state $\ket{\beta}\in B$ with $\ket{\beta}\otimes A^\dag < K$;
    \item $\|P_B\|_{\infty} < |A|$;
    \item $\tr_A (\1-P_{AB})$ has full rank;
    \item $\Aram(K)>1$.
  \end{enumerate}
  \renewcommand{\theenumi}{\arabic{enumi}}
\end{proposition}
\begin{proof}
\emph{$\overline{C}_{0EF}(K) > 0 \ \Rightarrow$ i.} has been shown in 
  the first part (necessity) of Theorem~\ref{thm:C-0EF-feasibility}, 
  at the start of this section, likewise \emph{i. $\Leftrightarrow$ ii.}.

\emph{ii. $\Leftrightarrow$ iii.} $P_{AB} \leq \1_A\ox\1_B$, hence
  $P_B \leq |A|\1_B$, i.e.~$\|P_B\|_\infty \leq |A|$. Equality is attained
  if and only if there exists an eigenvector $\ket{\beta}$ of $P_B$
  with eigenvalue $|A|$, which is equivalent to
  $|A| = \tr \proj{\beta}P_B = \tr (\1_A\ox\proj{\beta})P_{AB}$.
  But since $\1_A\ox\proj{\beta}$ has trace $|A|$ and $P_{AB}$
  is a projector, this is equivalent to $\1_A\ox\proj{\beta} \leq P_{AB}$,
  or again equivalently $\ket{\beta}\ox A^\dag < K$.

\emph{iii. $\Leftrightarrow$ iv.} $\|P_B\|_{\infty} < |A|$ if and only if
  $P_B = \tr_A P_{AB} < |A|\1_B$, if and only if 
  $\tr_A (\1-P_{AB}) > 0$.

\emph{iii. $\Rightarrow$ v.} Simply observe that $S= \frac{1}{\|P_B\|_\infty}\1_A$
  is feasible for $\Aram(K)$, since $\tr_A (S\ox\1)P_AB = \frac{1}{\|P_B\|_\infty}P_B \leq \1$,
  hence $\Aram(K) \geq \tr S = \frac{|A|}{\|P_B\|_\infty} > 1$.

\emph{v. $\Rightarrow$ ii.} We show the contrapositive: If $\ket{\beta}\otimes A^\dag < K$,
  then $\1\ox\proj{\beta} \leq P_{AB}$. Now, if $S$ is feasible for
  $\Aram(K)$, we have
  $\1_B \geq \tr_A(S\ox\1)P_{AB} \geq \tr_A(S\ox\1)(\1\ox\proj{\beta}) = (\tr S)\proj{\beta}$,
  hence $\tr S \leq 1$, and so $\Aram(K) = 1$.

\emph{iii. $\Rightarrow\ \overline{C}_{0EF}(K) > 0$.} Consider
  the cq-channel $\cM$ in eq.~(\ref{eq:M-channel}). It has output
  state support projectors
  \[
    P_{uv} = (X^u Z^v \ox \1) P_{AB} (Z^{-v} X^{-u} \ox \1),\quad u,v=1,\ldots,a,
  \]
  and we can verify directly that
  $\sum_{uv} P_{uv} = |A|\1_A\ox P_B$, so its norm satisfies
  $$\left\|\sum_{u,v} P_{uv} \right\|_\infty = |A|\, \|P_B\|_\infty < |A|^2.$$
  In other words, it satisfies the requirements of item iv) in Theorem ~\ref{exclude-discrimination},
  hence $\overline{C}_{0EF}(K) \geq \overline{C}_{0EF}(\cM) > 0$.
\end{proof}

\section{Shannon theoretic upper bound on $\mathbf{\overline{C}_{0EF}(K)}$}
\label{sec:C_min_E}
In this section we will develop an upper bound on the feedback-assisted
zero-error capacity via information theoretic ideas. For this
purpose we first review the classical case, due to Shannon.

\subsection{Shannon theoretic characterization of the fractional packing number:
\protect\\ Shannon's Conjecture}
\label{subsec:Shannon-conjecture}
The following characterization of the feedback-assisted zero-error capacity
of a classical channel was conjectured by Shannon at the end of
his seminal paper~\cite{Shannon56}, and to our knowledge proved first by 
Ahlswede~\cite{Ahlswede:feedback}, in the context of his treatment of the capacity 
of arbitrarily varying (classical) channels with instantaneous feedback, and 
using his very general results in that theory. 
Our proof seems more direct, but then it is specially geared towards the 
zero-error setting.

\begin{proposition}
  \label{prop:Shannon}
  For a bipartite graph $\Gamma$ on $\cX \times \cY$ such that 
  every $x \in \cX$ is adjacent to at least one $y \in \cY$,
  \[
    \log \alpha^*(\Gamma) = C_{\min}(\Gamma)
                         := \min \{ C(N) : \Gamma(N) \subseteq \Gamma \},
  \]
  where $C(N)$ is the usual Shannon capacity of a noisy classical 
  channel~\cite{Shannon48}.
\end{proposition}
\begin{proof}
The left hand side is the zero-error capacity of $\Gamma$, assisted
by feedback (plus some finite amount of communication), 
$\overline{C}_{0F}(\Gamma)$~\cite{Shannon56}.
From this, and the fact that feedback
does not increase the Shannon capacity of a channel~\cite{Shannon56}
(which may also be proved invoking the Reverse Shannon Theorem~\cite{BSST}),
it follows that $C(N) \geq \log \alpha^*(\Gamma)$
for any eligible $N$, hence $C_{\min}(\Gamma) \geq \log \alpha^*(\Gamma)$.

There is also a direct proof of this that avoids operational arguments,
relying instead only on elementary combinatorial notions. It goes via showing
that for every eligible channel $N$ and input probability distribution $p$,
\begin{equation}
  \label{eq:V}
  V(p) := \log \min_y \frac{1}{\sum_x \Gamma(y|x)p_x} \leq I(X:Y),
\end{equation}
which is enough because $\max_p V(p) = \log\alpha^*(\Gamma)$, 
while of course the maximum of $I(X:Y)$ equals $C(N)$.
Now, eq.~(\ref{eq:V}) is easily seen to be true for uniform distribution
$p_x = \frac{1}{|\cX|}$. Namely, with the equivocation sets 
$\cE_y = \{ x : \Gamma(y|x)=1 \}$
and the output probability distribution $q_y = \sum_x p_x N(y|x)$:
\[\begin{split}
  V(p) &=    \log |\cX| - \max_y \log |\cE_y| \\
       &\leq \log |\cX| - \sum_y q_y \log |\cE_y| \\
       &\leq \log |\cX| - \sum_y q_y H(X|Y=y)        \\
       &=    H(X) - H(X|Y) = I(X:Y),
\end{split}\]
where we have used the fact that $P_{X|Y=y}$ is supported
on $\cE_y$, and the uniformity of the distribution of $X$.
For non-uniform $p$, we use the method of types~\cite{Csiszar:method-of-types}
to reduce to the uniform case. In detail, consider the product distribution $p^{\ox n}$
and $X^n \sim p^{\ox n}$ as input to the i.i.d.~channel $N^{\ox n}$.
Introducing the type $T=T(X^n)$ of the string $X^n$, we have:
\[
  n I(X:Y) = I(X^n:Y^n) = I(T X^n:Y^n) = I(T:Y^n) + I(X^n:Y^n|T).
\]
On the other hand, for every type $\tau$,
\[\begin{split}
  2^{-n{V(p)}} &=    2^{-V(p^{\ox n})}                            \\
                   &=    \max_{y^n} \sum_{x^n} \Gamma(y^n|x^n) p_{x^n}      \\
                   &\geq \max_{y^n} \sum_{x^n\in\tau} \Gamma(y^n|x^n) p_{x^n} \\
                   &=    p^{\ox n}(\tau) \max_{y^n} \sum_{x^n\in\tau} \frac{1}{|\tau|}\Gamma(y^n|x^n),
\end{split}\]
since conditioned on $T(X^n)=\tau$, $X^n \sim u_\tau$ is uniformly distributed.
Hence, using the uniform case of the inequality (\ref{eq:V}),
\[
  n V(P) \leq \log\frac{1}{p^{\ox n}(\tau)} + V(u_\tau)
              \leq \log\frac{1}{p^{\ox n}(\tau)} + I(X^n:Y^n|T=\tau),
\]
and averaging over the different types this gives
\[
  n V(P) \leq H(T) + I(X^n:Y^n|T)
              \leq O(\log n) + n I(X:Y),
\]
because there are only $\text{poly}(n)$ many different types, 
and letting $n\rightarrow\infty$ we are done.

\medskip
So it remains only to show the opposite inequality. The proof
uses the primal and dual linear programming~\cite{Chvatal:LP} 
(LP) characterisations of $\alpha^*(\Gamma)$ to
construct an optimal channel $N(y|x)$, and in fact also an optimal
input distribution $p_x$, such that $C(N) = I(X:Y) = \log \alpha^*(\Gamma)$.

Recall the fractional packing number, eq.~(\ref{eq:fractional-packing-n}),
and choose optimal primal and dual solutions. 
Define an input distribution $p_x := \frac{w_x}{\alpha^*(\Gamma)}$.
This is the one that appears in Shannon's~\cite[Thm.~7]{Shannon56},
and his $\frac{1}{P_0}$ is the same as $\alpha^*(\Gamma)$.
Now, by complementary slackness~\cite{Chvatal:LP},
if $M_x := \sum_y \Gamma(y|x)v_y > 1$,
then $w_x = p_x = 0$; per contrapositive, if $p_x > 0$, then
$M_x = \sum_y \Gamma(y|x)v_y = 1$. Hence, we can define, for these
latter $x$,
\[
  N(y|x) := \Gamma(y|x)v_y,
\]
and in general for all $x$,
\[
  N(y|x) := \frac{1}{M_x}\Gamma(y|x)v_y.
\]

This is our candidate channel, and we have to convince ourselves that
indeed $C(N) = \log \alpha^*(\Gamma)$. First of all, let's confirm that with
the above distribution $p$, the mutual information $I(X:Y)$ equals
$\log \alpha^*(\Gamma)$. Let $D(p\|q)=\sum_x p(x)\log \frac{p(x)}{q(x)}$ be the relative entropy between two probability distributions $\{p_x\}$ and $\{q_x\}$, cf.~\cite{CoverThomas}.  Recall $I(X:Y) = \sum_x p_x D(N(\cdot|x)\|q)$,
with the output distribution
\[
  q_y = \sum_x p_x N(y|x)
      = \sum_x p_x \Gamma(y|x) v_y
      = \frac{v_y}{\alpha^*(\Gamma)},
\]
using once more complementary slackness: the equality is
trivial if $v_y = 0$, and if $v_y > 0$ then $\sum_x \Gamma(y|x)w_x = 1$.
In the present case, we calculate for all $x$,
\[
  D\bigl(N(\cdot|x)\|q\bigr)
   = \sum_y \frac{\Gamma(y|x)v_y}{M_x} \log\frac{\Gamma(y|x)\frac{v_y}{M_x}}{\frac{v_y}{\alpha^*(\Gamma)}}
   = \log \frac{\alpha^*(\Gamma)}{M_x},
\]
which is $\log \alpha^*(\Gamma)$ for all $p_x>0$ as then $M_x=1$.
So indeed $I(X:Y) = \log \alpha^*(\Gamma)$. But we see even more:
While all the relative entropies $D(N(\cdot|x)\|q)$ with $p_x > 0$ are
equal to $\log \alpha^*(\Gamma)$, for $p_x = 0$ instead,
\[
  D\bigl(N(\cdot|x)\|q\bigr)
               =    \log \frac{\alpha^*(\Gamma)}{M_x}
               \leq \log \alpha^*(\Gamma),
\]
because $M_x \geq 1$. These two conditions (for $p_x>0$ and $p_x=0$)
are well known, classic characterizations of the Shannon capacity 
(cf.~\cite{Csiszar:relent,SW:optimal-ens});
they characterize an optimal input distribution for given channel
$N$, so indeed we prove $C(N) = \log \alpha^*(\Gamma)$.
\end{proof}

\medskip\noindent
{\bf Remark.}
Note that neither is $C_{\min}$ altered by allowing the use of 
entanglement as well as feedback~\cite{BSST},
nor $C_{0F}$ by allowing the use of entanglement and other
no-signalling correlations~\cite{CLMW:0}.

\subsection{Quantum generalization of the Shannon bound}
\label{subsec:quantum-Shannon-conjecture}
Recall that for a channel $\cN:\cS(A) \longrightarrow \cS(B)$,
the entanglement-assisted classical capacity~\cite{BSST},
i.e.~the maximum rate of asymptotically error-free communication
via many uses of the channel assisted by a suitable pre-shared
entangled state, is given by
\[\begin{split}
  C_E(\cN) = \max_{\rho} I(A:B)_\sigma
           = \max_{\rho} \left\{ S(\rho) + S({\cal N}(\rho))
                                   - S\bigl( (\id\otimes{\cal N})\phi \bigr) \right\},
\end{split}\]
where $\sigma_{AB} = (\id \ox \cN)\phi_{AA'}$ is the joint
input-output state, $\phi_{AA'}$ is a purification of $\rho$,
and $I(A:B) = S(\sigma_A)+S(\sigma_B)-S(\sigma_{AB})$
is the quantum mutual information. In the particular
case above, we also write it 
$I(\rho;{\cal N}) = S(\rho) + S({\cal N}(\rho)) - S\bigl( (\id\otimes{\cal N})\phi \bigr)$.

Using this, we define for a non-commutative bipartite graph 
$K < \cL(A\rightarrow B)$ such that $\1 \in K^\dagger K$ 
(these are precisely the possible Kraus subspaces of channels):
\[
  C_{\min E}(K) := \min \{ C_E(\cN) : \cK(\cN) < K \}.
\]
That this is indeed a minimum follows from continuity of $C_E$ and
the fact that the eligible channels form a compact convex set.
This definition is of course motivated by Proposition~\ref{prop:Shannon},
suggesting $2^{C_{\min E}(K)}$ as a possible quantum
generalisation of the fractional packing number. For one thing,
for the quantum realisation $K$ of a classical equivocation graph
$\Gamma$, it is easy to see that indeed
$C_{\min E}(K) = C_{\min}(\Gamma) = \log \alpha^*(\Gamma)$,
see the remark at the end of the preceeding 
Subsection~\ref{subsec:Shannon-conjecture}.

At least, this quantity is related to the feedback-assisted zero-error
capacity: Indeed, the result of Bowen~\cite{Bowen:feedback}
(alternatively the Quantum Reverse Shannon Theorem~\cite{QRST,BertaChristandlRenner}) 
tells us that $C_E(\cN)$ is not increased even by allowing feedback,
so that $C_{0EF}(K)$ (and actually even $\overline{C}_{0EF}(K)$) is upper bounded by
the entanglement-assisted capacity $C_E({\cal N})$ for any
channel ${\cal N}$ such that $\cK({\cal N})<K$, hence
\begin{theorem}
  \label{thm:upper-bound}
  $\overline{C}_{0EF}(K) \leq C_{\min E}(K)$ for any non-commutative 
  bipartite graph $K < \cL(A\rightarrow B)$.
  \qed
\end{theorem}

$C_{\min E}(K)$ shares many properties with $C_{\min}(\Gamma)$, to
which it reduces for classical channels.
First, $C_{\min E}(K)$ is given by a minimax
formula (min over channels and max over quantum mutual information
-- see below) to which the minimax theorem applies, so it is also
given by a maximin (Lemma~\ref{lemma:minimax} below). Second, using
this characterisation and properties of the von Neumann entropy,
it can be shown that $C_{\min E}$ is additive
(Lemma~\ref{lemma:add} below). Third, thanks to the operational
definition of $C_{E}$, it can be easily seen to be monotonic
under pre- and post-processing (Lemma~\ref{lemma:mono} below).

We shall need some well-known mathematical properties of the
quantum mutual information. The first is 
that $I(\rho;{\cal N})$ is concave in $\rho$ and convex in ${\cal N}$,
just like its classical counterpart~\cite{AdamiCerf,BSST}.
The convexity in ${\cal N}$ follows from strong subadditivity: Let
\begin{align*}
  \sigma^{AB} &= \bigl( \id\otimes(\lambda{\cal N}^{(1)}+(1-\lambda){\cal N}^{(2)}) \bigr)\phi_\rho \\
              &= \lambda\sigma^{(1)}_{AB} + (1-\lambda)\sigma^{(2)}_{AB}   \\
              &= \tr_{B'} \widetilde{\sigma}_{ABB'},
\end{align*}
with
$\widetilde{\sigma}_{ABB'} = \lambda\sigma_{AB}^{(1)}\otimes\proj{1}_{B'}
                                + (1-\lambda)\sigma_{AB}^{(2)}\otimes\proj{2}_{B'}.$
Then,
\[\begin{split}
  I\bigl(\rho;\lambda{\cal N}^{(1)}+(1-\lambda){\cal N}^{(2)}\bigr)
        &=    I(A:B)_\sigma                  \\
        &\leq I(A:BB')_{\widetilde{\sigma}}  \\
        &=    \lambda I(A:B)_{\sigma^{(1)}} + (1-\lambda) I(A:B)_{\sigma^{(2)}} \\
        &=    \lambda I(\rho;{\cal N}^{(1)}) + (1-\lambda) I(\rho;{\cal N}^{(2)}).
\end{split}\]

The concavity in $\rho$ can be seen as follows, using strong
subadditivity again: For states $\rho^{(1)}$, $\rho^{(2)}$ with purifications
$\phi^{(1)}$, $\phi^{(2)}$, respectively, and $0\leq \lambda\leq 1$, we
construct a purification of the mixture $\lambda\rho^{(1)}+(1-\lambda)\rho^{(2)}$,
as follows:
\[
  \ket{\phi} = \sqrt{\lambda}\ket{\phi^{(1)}}\ket{11}_{A'A''} + \sqrt{1-\lambda}\ket{\phi^{(2)}}\ket{22}_{A'A''}.
\]
With $\sigma_{AA'A''B} = (\id_{AA'A''}\otimes{\cal N})\phi$, we have
\[\begin{split}
  I\bigl( \lambda\rho^{(1)}+(1-\lambda)\rho^{(2)};{\cal N} \bigr)
        &=    I(AA'A'':B)_\sigma  \\
        &\geq I(AA':B)_\sigma  \\
        &\geq I(A:B|A')_\sigma \\
        &=    \lambda I(\rho^{(1)};{\cal N}) + (1-\lambda) I(\rho^{(2)};{\cal N}).
\end{split}\]

\begin{lemma}
  \label{lemma:minimax}
  For any non-commutative bipartite graph
  $K < {\cal L}(A\!\rightarrow\! B)$,
  \[\begin{split}
    C_{\min E}(K) 
         &= \min_{{\cal N} \text{ s.t.} \atop \cK({\cal N}) < K} \max_\rho\ \ \,I(\rho;{\cal N}) \\
         &= \ \ \,\max_\rho \min_{{\cal N} \text{ s.t.} \atop \cK({\cal N}) < K} I(\rho;{\cal N}).
  \end{split}\]
\end{lemma}
\begin{proof}
The first equation is the definition of $C_{\min E}(K)$, with
the formula for $C_E({\cal N})$ inserted. Above we saw that the
argument $I(\rho;{\cal N})$ is concave in the first and convex
in the second argument. Hence von Neumann's minimax theorem, or
rather its generalisation due to Sion~\cite{Sion} applies, allowing us
to interchange the order of min and max.
\end{proof}

\begin{lemma}
  \label{lemma:add}
  For non-commutative bipartite graphs $K_1 < {\cal L}(A_1\!\rightarrow\! B_1)$
  and $K_2 < {\cal L}(A_2\!\rightarrow\! B_2)$,
  \[
    C_{\min E}(K_1 \otimes K_2) = C_{\min E}(K_1) + C_{\min E}(K_2).
  \]
\end{lemma}
\begin{proof}
We show this by separately demonstrating ``$\leq$'' and ``$\geq$''
in the above relation, using the two expressions for $C_{\min E}$
from Lemma~\ref{lemma:minimax}. In the following, choose optimal
states $\rho_1$, $\rho_2$ and channels ${\cal N}_1$, ${\cal N}_2$
for $K_1$, $K_2$, respectively.

\medskip\noindent
{\bf ``$\mathbf{\leq}$'':} By the first expression in
Lemma~\ref{lemma:minimax},
\[\begin{split}
  C_{\min E}(K_1\otimes K_2) &\leq \max_\rho I(\rho;{\cal N}_1\otimes{\cal N}_2) \\
                             &=    C_E({\cal N}_1\otimes{\cal N}_2) \\
                             &=    C_E({\cal N}_1) + C_E({\cal N}_2) \\
                             &=    C_{\min E}(K_1) + C_{\min E}(K_2),
\end{split}\]
using the fact that the entanglement-assisted capacity is additive,
proved by Adami and Cerf in~\cite{AdamiCerf}. Note that
$\cK({\cal N}_1\otimes{\cal N}_2) = \cK({\cal N}_1) \otimes \cK({\cal N}_2)< K_1 \otimes K_2$.

\medskip\noindent
{\bf ``$\mathbf{\geq}$'':} By the second expression in
Lemma~\ref{lemma:minimax},
\[
  C_{\min E}(K_1\otimes K_2)
    \geq \min_{{\cal N} \text{ s.t.}\atop \cK({\cal N}) < K_1\otimes K_2}
                                                  I(\rho_1\otimes\rho_2;{\cal N}),
\]
and we need only to show that the minimum is attained at a
product channel ${\cal N} = {\cal N}_1\otimes{\cal N}_2$
with $\cK({\cal N}_i)< K_i$.
For this purpose, consider the state
\[
  \sigma_{A_1A_2B_1B_2} = (\id_{A_1}\otimes\id_{A_2}\otimes{\cal N})(\phi_{1}\otimes\phi_{2}),
\]
for the purifications $\phi_i$ of $\rho_i$ ($i=1,2$). Now observe
that with respect to $\sigma$,
\[\begin{split}
  I(A_1A_2:B_1B_2) - I(A_1:B_1) - I(A_2:B_2)
          &= S(A_1A_2)+S(B_1B_2)-S(A_1A_2B_1B_2) \\
          &\phantom{======} -S(A_1)-S(B_1)+S(A_1B_1) \\
          &\phantom{======} -S(A_2)-S(B_2)+S(A_2B_2) \\
          &= I(A_1B_1:A_2B_2) - I(B_1:B_2) - I(A_1:A_2) \geq 0,
\end{split}\]
because $I(A_1:A_2)=0$ and by strong subadditivity. In other words,
\[\begin{split} 
  I(A_1A_2:B_1B_2)_\sigma &\geq I(A_1:B_1)_{\sigma_1} + I(A_2:B_2)_{\sigma_2} \\
                          &=    I(A_1A_2:B_1B_2)_{\sigma_1\otimes\sigma_2},
\end{split}\]
with the reduced states
\begin{align*}
  \sigma_1 = \sigma_{A_1B_1} &= \tr_{A_2B_2} \sigma
                              = \bigl(\id_{A_1}\otimes(\tr_{B_2}\!\circ{\cal N})\bigr)(\phi_1\otimes\rho_2), \\
  \sigma_2 = \sigma_{A_2B_2} &= \tr_{A_1B_1} \sigma
                              = \bigl(\id_{A_2}\otimes(\tr_{B_1}\!\circ{\cal N})\bigr)(\rho_1\otimes\phi_2).
\end{align*}
I.e.,
\[
  I(\rho_1\otimes\rho_2;{\cal N}) \geq I\bigl(\rho_1;\tr_{B_2}\!\circ{\cal N}(\cdot\otimes\rho_2)\bigr)
                                        + I\bigl(\rho_2;\tr_{B_1}\!\circ{\cal N}(\rho_1\otimes\cdot)\bigr).
\]

Finally, $\tr_{B_2}\!\circ{\cal N}(\cdot\otimes\rho_2)$ is eligible:
If ${\cal N}$ has Kraus operators $E_i \in K_1\otimes K_2 < {\cal L}(A_1A_2\!\rightarrow\!B_1B_2)$,
and choosing an eigenbasis of $\rho_2$ and an arbitrary basis of $B_2$,
\[
  \cK\bigl(\tr_{B_2}\!\circ{\cal N}(\cdot\otimes\rho_2)\bigr)
     = \operatorname{span}\, \bigl\{ \bra{j}_{B_2} E_i \ket{k}_{A_2} : i,j,k \bigr\}
     < K_1.
\]
$\cK\bigl(\tr_{B_1}\!\circ{\cal N}(\rho_1\otimes\cdot)\bigr) < K_2$ is
analogous, and we are done.
\end{proof}
%

\begin{lemma}
  \label{lemma:mono}
  All of $C_{0EF}$, $\overline{C}_{0EF}$ and $C_{\min E}$ are monotonic under 
  pre- and post-processing of the channel: For non-commutative bipartite graphs
  $K < \cL(A\rightarrow B)$ and $K_A < \cL(U \rightarrow A)$,
  $K_B < \cL(B \rightarrow V)$, the matrix-multiplied space 
  $K_B K K_A < \cL(U \rightarrow V)$ is a non-commutative bipartite
  graph, and
  \begin{align*}
    C_{0EF}(K)            &\geq C_{0EF}(K_B K K_A),            \\
    \overline{C}_{0EF}(K) &\geq \overline{C}_{0EF}(K_B K K_A), \\
    C_{\min E}(K)         &\geq C_{\min E}(K_B K K_A).
  \end{align*}
\end{lemma}
\begin{proof}
  For $C_{0EF}$ and $\overline{C}_{0EF}$ this follows directly from 
  the operational definition:
  the pre- and post-processings may be absorbed into
  the input modulation and feedback-decoding, respectively,
  showing that a zero-error code for $K_B K K_A$ yields one
  for $K$.
  
  For $C_{\min E}$, the argument is similar using the fact that
  $C_E(\cN)$ is the operational entanglement-assisted capacity
  of the channel $\cN$~\cite{BSST}.
\end{proof}

\medskip
We can now give yet another characterization of the feasibility of
$\overline{C}_{0EF}(K)>0$, adding to the list of Theorem~\ref{thm:C-0EF-feasibility}
and Proposition~\ref{prop:feasibility-equiv}.

\begin{theorem}
  \label{thm:feasibility-C-minE}
  For any non-commutative bipartite graph $K$,
  $\overline{C}_{0EF}(K)>0$ if and only if $C_{\min E}(K)>0$. 
\end{theorem}
\begin{proof}
The only way in which $C_{\min E}(K)$ can be $0$ is that
there is a channel $\mathcal{N}$ with $\mathcal{K}(\mathcal{N}) < K$
and $C_E(\mathcal{N})=0$, i.e.~$\mathcal{N}$ has to be constant.
We have seen that this is eqivalent to $\ket{\beta}\ox A < K$
for a state vector $\ket{\beta}\in B$. But by Theorem~\ref{thm:C-0EF-feasibility}
this is precisely the characterization of $\overline{C}_{0EF}(K)$
being $0$.
\end{proof}

\bigskip
To illustrate the bound of Theorem~\ref{thm:upper-bound}, 
we consider the example of Weyl diagonal channels and the
dependence on the output state geometry for cq-channels.

\medskip\noindent
{\bf Weyl diagonal channels.}
Denoting by $X$ and $Z$ the 
discrete translation and phase shift (which generate a subgroup of
the unitary group of cardinality $d^3$, thanks to the commutation
relation $XZ = \omega ZX$, $\omega = e^{2\pi i/d}$), consider the
channel
\[
  {\cal N}(\rho) = \sum_{a,b=0}^{d-1} p_{ab} X^a Z^b \rho Z^{-b} X^{-a},
\]
with probabilities $p_{ab} \geq 0$ summing to $1$. Clearly,
\[
  \cK({\cal N}) = \operatorname{span}\{ W_{ab} := X^a Z^b \,:\, p_{ab} > 0 \},
\]
i.e.~this $K$ is characterised by a subset ${\cal S} \subset \ZZ_d \times \ZZ_d$.
It supports precisely those Weyl diagonal channels ${\cal N}$ with $p_{ab}=0$
for $ab \not\in {\cal S}$ -- and of course many channels that are
not Weyl diagonal.

First, note that ${\cal N}$ above is Weyl-covariant:
\[
  {\cal N}(W_{ab}\rho W_{ab}^\dagger) = W_{ab} {\cal N}(\rho) W_{ab}^\dagger
\]
for all $ab$. From this, and the irreducibility of the action of
the Weyl operators on $\CC^d$, it follows that
\[
  C_E({\cal N}) = I\left(\frac{1}{d}\1;{\cal N}\right)
                  = 2 \log d - H(\vec{p}),
\]
where $\vec{p}=(p_{ab}: a,b=0,\cdots, d-1)$ is the probability vector.
This means that for a $k$-element ${\cal S} \subset \ZZ_d\times \ZZ_d$
and $K = \operatorname{span}\{ W_{ab} \,:\, ab \in {\cal S} \}$,
\begin{equation}
  \label{eq:Weyl-diag-CE}
  \min_{{\cal N} \text{ Weyl-diag.} \atop \cK({\cal N}) < K} C_E({\cal N}) = 2\log d - \log k,
\end{equation}
the minimum being attained at the uniform distribution on ${\cal S}$:
$p_{ab} = \frac{1}{k}$ for $ab\in{\cal S}$, and $0$ otherwise.

We will now show that $2\log d - \log k$ is an achievable rate of
zero-error communication via this channel when assisted by feedback
(plus a constant activating amount of noiseless communication).
The key is the observation that if we use
\[
  {\cal N}_0(\rho) = \frac{1}{k}\sum_{ab\in{\cal S}} W_{ab} \rho W_{ab}^\dagger
\]
with dense coding, i.e.~with a
maximally entangled state $\ket{\Phi_d}$ and sender modulation
by the very Weyl operators $W_{ab}$, the receiver making a Bell
measurement in the basis $(W_{ab}\otimes\1)\ket{\Phi_d}$, we obtain a
generalised typewriter channel
\begin{align*}
  T : \ZZ_d\times\ZZ_d &\longrightarrow \ZZ_d\times\ZZ_d, \\
              T(ab|cd) &= \begin{cases}
                            \frac{1}{k} & \text{ if } (a-c,b-d)\in{\cal S}, \\
                            0           & \text{ otherwise.}
                          \end{cases}
\end{align*}
(And choosing a different ${\cal N}$ supported by $K$ changes only
the non-zero transition probabilities.)
$T$ is easily seen to have fractional packing number $d^2/k$,
so its activated feedback-assisted zero-error capacity is $2\log d - \log k$.
Hence $\overline{C}_{0EF}(K) \geq 2\log d - \log k$, and together with
eq.~(\ref{eq:Weyl-diag-CE}), we conclude
\[
  \overline{C}_{0EF}(K) = C_{\min E}(K) = \overline{C}_{0F}(T) = 2\log d - \log k.
\]

Finally, this is also the minimal zero-error communication cost
to simulate a channel supported by $K$ (using entanglement and
shared randomness), making use of an idea in~\cite{BSST}:
By the results of~\cite{CLMW:0}, one can simulate $T$ with free
shared randomness at communication rate $2\log d - \log k$. Now,
if in the teleportation protocol using a maximally entangled state
and the Weyl unitaries $W_{ab}$, we replace the noiseless channel
of $d^2$ messages by this $T$, one simulates exactly ${\cal N}_0$.
\qed

\medskip\noindent
{\bf Nontrivial dependence of $\mathbf{\overline{C}_{0EF}}$ on the channel geometry.}
Consider a non-commutative bipartite graph corresponding to a
pure state cq-channel, $K = \operatorname{span}\{ \ketbra{\psi_i}{i} \}$. 
We can see that $\overline{C}_{0EF}(K)$ depends nontrivially 
on the geometry of the vector arrangement of the $\ket{\psi_i}$, 
even if they are all pairwise non-orthogonal: Indeed,
when they are close to parallel, $\overline{C}_{0EF}(K)$ is arbitrarily
close to $0$, but when they are sufficiently close to being mutually orthogonal, 
$\overline{C}_{0EF}(K)$ is arbitrarily close to $\log|A|$.

Clearly, the closer to being parallel the $\ket{\psi_i}$ are,
the larger the required $n$ in the argument in Subsection~\ref{subsec:pure-cq}
becomes, so the lower bound moves closer to $0$. 
On the other hand, this is really necessary, since
\[
  C_{\min E}(K) = \max_{(p_i)} S\left(\sum_i p_i \proj{\psi_i}\right)
\]
converges to $0$ as the $\ket{\psi_i}$ get closer to being collinear. 

In the other extreme, to show that $C_{0EF}(K) \rightarrow \log |A|$
when $C_{\min E}(K) \rightarrow \log |A|$,
i.e.~when the $\psi_i$ become closer and closer to being orthogonal,
we use once more the ideas from Subsection~\ref{subsec:pure-cq}:
Assume that for all $i\neq j$, $|\bra{\psi_i}\psi_j\rangle| \leq \epsilon$,
which is a more convenient expression for $C_{\min E}(K) \geq \log |A| - \delta$.

We claim that if $\epsilon$ is small enough, we can use $K$ to simulate 
a ``random superset channel'' (cf.~\cite{CLMW:0}): 
for integers $t < a = |A|$ define the classical channel 
$S_{1,t}^a : [a] \rightarrow {[a] \choose t}$ such that
\[
  S_{1,t}^a: [a] \ni i \longmapsto J \in {[a] \choose t} \text{ randomly with } i \in J,
\]
where ${[a] \choose t}=\{J: J\subseteq [a], |J|=t\}$, the collection of all subsets of $[a]$ with $t$ elements. Note that the transition probability matrix of $S_{1,t}^a$ is given by $\{p(J|i)\}$ such that 
$$p(J|i)={{a-1 \choose t-1}}^{-1},~i\in [a],~J\in {[a] \choose t}.$$

Indeed, we use the characterization of~\cite{CJW-pure-trans},
which will show that there is a deterministic transformation of the 
set $\{\ket{\psi_i}\}$ to the set $\{\ket{\varphi_i}\}$, with
\[
  \ket{\varphi_i} = \frac{1}{\sqrt{{a-1 \choose t-1}}} \sum_{i \in J \in {[a]\choose t}} \ket{J}
                  \in \CC^{{a \choose t}}.
\]
Once this is achieved, Bob measures the states $\ket{\varphi_i}$
in the computational basis, resulting in an output of the channel
$S_{1,t}^a$.
To see this in detail, let us focus on the smallest possible case $t=2$,
for which we see that for $i\neq j$, $\bra{\varphi_i}\varphi_j\rangle = \frac{1}{a-1}$.
The necessary and sufficient condition required in~\cite{CJW-pure-trans}
for the existence of a cptp map transforming $\{\ket{\psi_i}\}$ into 
$\{\ket{\varphi_i}\}$ is that there exists a positive semidefinite $a\times a$-matrix
$M$ such that $\Psi = \Phi \circ M$, where $\Psi = \bigl[ \bra{\psi_i}\psi_j\rangle \bigr]$
and $\Phi = \bigl[ \bra{\varphi_i}\varphi_j\rangle \bigr]$ are the Gram matrices
of the two input/output state sets, and $\circ$ denotes the elementwise
(Hadamard/Schur) product. In other words, 
\[
  M = \Psi \circ \Phi^{\circ -1} \geq 0, \quad\text{i.e.}\quad (a-1)\Psi \geq (a-2)\1.
\]
However, all eigenvalues of $\Psi$ are lower bounded by $1-(a-1)\epsilon$, 
which is $\geq \frac{a-2}{a-1}$ as soon as $\epsilon \leq \frac{1}{(a-1)^2}$.
In this case, we find 
$\overline{C}_{0EF}(K) \geq \overline{C}_{0F}(S_{1,2}^a) = \log a - 1$.
Applying the same to multiple copies of the channel, this reasoning
shows that if $\epsilon \leq (|A|-1)^{-2n}$, then
$\overline{C}_{0EF}(K) \geq \log a - \frac1n$.
\qed

\bigskip
We do not know whether in general $\overline{C}_{0EF}$ equals
$C_{\min E}$ or not. However, we can show that the latter is
a genuine capacity, as per the following theorem, whose proof
however we relegate to Appendix~\ref{app:adversarial} because
it would detract from our principal, zero-error argument.

\begin{theorem}
  \label{thm:adversarial}
  For any non-commutative bipartite graph $K$, the
  \emph{adversarial entanglement-assisted classical capacity}
  of $K$ is given by $C_{\ast E}(K) = C_{\min E}(K)$. 
\end{theorem}

The definition of this capacity is as follows: An entanglement-assisted
$n$-block code consists of an entangled state (w.l.o.g.~pure) 
$\ket{\phi}^{A_0 B_0}$, $N$ modulation cptp maps
$\mathcal{E}_i:\mathcal{L}(A_0) \rightarrow \mathcal{L}(A^n)$
($m=1,\ldots,N$), and a POVM $(D_i)_{i=1}^N$ on $B_0 B^n$.
The code is said to have error $\epsilon$ for $K^{\ox n}$
if the (average) error probability,
\[
  P_{\text{err}}\bigl(\cN^{(n)}\bigr)
       = \frac{1}{N}\sum_{i=1}^N \Bigl( 1 - \tr\bigl( (\cN^{(n)}\circ\cE_i \ox \id)\phi\bigr) D_i \Bigr),
\]
is $\leq \epsilon$ for every channel $\cN^{(n)}$ with 
$\mathcal{K}(\cN^{(n)}) < K^{\ox n}$. In this case, we call the 
collection $(\phi;\cE_i,D_i)$ an $(n,\epsilon)$-code for $K^{\ox n}$.
Denoting the largest number $N$ of messages of an $(n,\epsilon)$-code
as $N(n,\epsilon;K)$, the \emph{adversarial entanglement-assisted classical capacity}
is defined as
\[
  C_{\ast E}(K) := \inf_{\epsilon > 0} \liminf_{n\rightarrow\infty} \frac1n \log N(n,\epsilon;K).
\]

In Appendix~\ref{app:adversarial} we shall actually show that 
\[
  \lim_{n\rightarrow\infty} \frac1n \log N(n,\epsilon;K) = C_{\min E}(K)
\]
for every $0 <\epsilon < 1$ (this is known as a strong converse).
There we will see that even allowing entanglement and arbitrary feedback
in the communication protocol does not increase the capacity 
$C_{\ast E}(K)$ beyond $C_{\min E}(K)$, hence we may also address it 
as \emph{feedback-assisted adversarial capacity} $C_{\ast EF}(K)$.

\section{Conclusion}
\label{sec:outro}
We have introduced the problem of determining the zero-error capacity
of a quantum channel assisted by noiseless feedback. We showed that the
capacity only depends on the ``non-commutative bipartite graph'' $K$ of the
channel, and that every nontrivial $K$ has positive capacity.

Motivated by Shannon's treatment of the classical case, we considered the
minimisation of entanglement-assisted classical capacities over all
channels with the same non-commutative bipartite graph
and proved several properties of this definition: it is an upper
bound on the activated feedback-assisted zero-error capacity, it is given by
a minimax/maximin formula, and is additive. It is also shown to be
equal to the adversarial entanglement-assisted capacity.

Note that when restricting all statements above to classical channels,
which are given by a bipartite equivocation graph $\Gamma$, all of these
quantities boil down to the fractional packing number:
\[
  2^{C_{\min E}(K)} = 2^{C_{\min}(\Gamma)} = 2^{\overline{C}_{0F}(\Gamma)} = \alpha^*(\Gamma),
\]
which furthermore quantifies the zero-error capacity and simulation 
cost of $\Gamma$ when assisted by general no-signalling correlations~\cite{CLMW:0},
$2^{C_{0,\text{NS}}(\Gamma)} = 2^{S_{0,\text{NS}}(\Gamma)} = \alpha^*(\Gamma)$.
However, for quantum channels and non-commutative bipartite graphs 
these notions start diverging, so none of them can be considered
as a preferred ``quantum fractional packing number'': 
In~\cite{DW:ns-ass}, no-signalling assisted zero-error
capacity and simulation cost were determined for cq-channels,
$C_{0,\text{NS}}(K) = \log \Aram(K)$ and $S_{0,\text{NS}}(K) = \log \Sigma(K)$,
with the semidefinite packing number $\Aram(K)$ and another SDP $\Sigma(K)$,
and while in general (for cq-channels)
\[
  \log \Aram(K) \leq C_{\min E}(K) \leq \log \Sigma(K),
\]
both inequalities can be strict~\cite{DW:ns-ass}. It remains an open
question how $\overline{C}_{0EF}(K)$ fits into this picture, and in particular
whether it is equal to or sometimes strictly smaller than $C_{\min E}(K)$.
We believe that pure state cq-channels offer a good testing ground
for ideas; we might take encouragement from~\cite{TKG:unambiguous}, where it 
was shown that the \emph{unambiguous capacity} of a pure state cq-graph $K$
equals $C_{\min E}(K)$. Other interesting $K$ are those that admit only
one channel $\cN$, for instance channels extremal in the set of cptp maps, 
cf.~\cite{DW:ns-ass}, an example of which is the amplitude damping channel;
in this case, $C_{\min E}(K) = C_E(\cN)$.

Next, motivated by the fact that both $\Aram(K)$ and $\Sigma(K)$ are SDPs
(at least for cq-graphs), we ask if there is a manifestly semidefinite programming
(or even just convex optimisation) characterisation of $2^{C_{\min E}(K)}$? 
To make progress, we need at least to understand some properties of an optimal 
${\cal N}$ for given $K$, and potentially also an optimal input state.

To offer a concrete approach to the question whether $C_{\min E}(K)$
is an achievable rate for pure state cq-graph $K$, we suggest to look
at the possible use of conclusive exclusion to implement a list-decoding
protocol, by excluding more than one state by each outcome
-- cf.~\cite{BJOP}.
 
\begin{quote}
\emph{List-decoding from approximate decoding?}
Given state vectors $\ket{\psi_1},\ldots,\ket{\psi_N} \in B$ (w.l.o.g.~$|B|=N$)
that are sufficiently orthogonal in the sense that there exists
an orthonormal basis $\{\ket{v_1},\ldots,\ket{v_N}\}$ of $B$ such that
\[
  \forall i \quad |\langle v_i \ket{\psi_i}|^2 \geq 1-\epsilon.
  \ 
  \left(\text{For instance, this holds if for each }
  i,\ \sum_{j\neq i} |\langle \psi_i \ket{\psi_j}|^2 \leq \epsilon,\
  \text{by~\cite{HausladenJozsaSchumacherWestmorelandWootters}.}
  \right)
\]
Then, does there exist a subset of $N' \geq \Omega(N^{1-\delta})$ of
these states, $\{ \ket{\psi_{i_j}} : j=1,\ldots N' \}$,
$L \leq O(N^\delta)$ 
($\delta \rightarrow 0$ with $\epsilon \rightarrow 0$ uniformly)
and a POVM $\left( M_S : S\in{[N']\choose L} \right)$, such that
$\bigl\{ j : \bra{\psi_{i_j}}M_S\ket{\psi_{i_j}} \neq 0 \bigr\} \subset S$
for all $S\in{[N']\choose L}$?
\end{quote}

Note that a positive answer would imply that by preparing 
$\psi_{i_j}$ and measuring the POVM elements
$M_S$, we construct a classical channel/hypergraph $\Gamma$ with
$\alpha^*(\Gamma) \geq \frac{N'}{L}$.
To see this, observe that each output $S$ is reached from at most
$L$ inputs $j$, namely those $j\in S$, so the weight distribution 
$w_j = \frac{1}{L}$ for all $i$ is admissible in the definition of 
$\alpha^*(\Gamma)$.
Thus we would obtain
\[
  \overline{C}_{0EF}(K) \geq \overline{C}_{0F}(\Gamma) 
                        \geq \log\frac{N'}{L} \geq (1-2\delta)\log N - O(1),
\]
which is at least consistent with $C({\cal N})$ being
of the order $(1-\epsilon)\log N - O(1)$, by the existence of
the basis $\{\ket{v_1},\ldots,\ket{v_N}\}$ and Fano's inequality.

By Hausladen et~al.~\cite{HausladenJozsaSchumacherWestmorelandWootters}
this would imply that we can asymptotically achieve the rate $C(\cN)=C_{\min E}(K)$
as activated feedback-assisted zero-error capacity, where
$K=\text{span}\{\ketbra{\psi_i}{i}:i=1,\ldots,N\}$.
It would also imply a new proof of the result of~\cite{TKG:unambiguous},
since we could use the Shannon scheme~\cite{Shannon56} to get arbitrarily
close to the rate $\log\alpha^*(\Gamma)$ by a deterministic list-decoding
with constant list size, and then constant activating communication, 
which we clearly can realize in an unambiguous fashion with constant
overhead.

Finally, there is another generalization of the instantaneous feedback
considered by Shannon, which was dubbed ``coherent feedback'' in~\cite{QRST},
and which consist in the channel environment $C$ from the Stinespring
isometry $V:A\hookrightarrow B \ox C$ to be handed back to Alice. 
More like Shannon's model, it is completely passive as it doesn't
involve any action of Bob's. The resulting zero-error capacity,
$C_{0\ket{F}}(V)$ is not even obviously a function of $K$ only, nor
is it clear whether additional free entanglement or free active
feedback from Bob to Alice will increase it, though it is clear
from the Quantum Reverse Shannon Theorem that all of $C_{0\ket{F}}(V)$ 
and its variants are upper bounded by $C_{E}(\cN)$.

\bigskip\noindent
{\bf Acknowledgments.}
We thank Marcin Paw\l{}owski and Ciara Morgan for listening with
empathy to our teething problems with the mixed state cq-channels,
Aram Harrow, Janis N\"otzel and many others for enjoyable discussions
on zero-error information theory, and Matthias Christandl for 
a conversation on post-selection lemmas. 
Special thanks to Will Matthews for pointing out an error in our 
first proof of Shannon's Conjecture. We are also grateful to an anonymous reviewer 
for pointing out to us that the second part of the proof of Lemma \ref{lemma:add} 
was presented in \cite{HsiehDattaWilde} to show the superadditivity of mutual information.

RD was or is supported in part by the Australian Research Council (ARC) under 
Grant DP120103776 (with AW), 
and by the National Natural Science Foundation of China under grant 
no. 61179030. He was also supported in part by an ARC 
Future Fellowship under Grant FT120100449.
SS was or is supported by a Newton International Fellowship, the Royal Society
and the U.K.~EPSRC.
AW was or is in part supported by the European Commission (STREPs ``QCS''
and ``RAQUEL''), the European Research Council (Advanced Grant ``IRQUAT''), 
the U.K.~EPSRC, the Royal Society and a Philip Leverhulme Prize. 
Furthermore, by the Spanish MINECO, projects FIS2008-01236 and
FIS2013-40627-P, with the support of FEDER funds, as well as by
the Generalitat de Catalunya CIRIT, project no.~2014 SGR 966.

\appendix

\section{$\mathbf{C_{\min E}(K)}$ equals the adversarial entanglement-assisted capacity}
\label{app:adversarial}

Here we give a complete proof of the following theorem from
Section~\ref{sec:C_min_E}.

\medskip\noindent
{\bf Theorem~\ref{thm:adversarial}}
For any non-commutative bipartite graph $K$, the
\emph{adversarial entanglement-assisted classical capacity}
of $K$ is given by $C_{\ast E}(K) = C_{\min E}(K)$. 

\bigskip
Before proving it, we show a simpler statement on so-called compound 
channels, which will be pivotal for the general proof, however. 
For a non-commutative bipartite graph $K < \cL(A\rightarrow B)$,
and a pure state $\ket{\phi} \in AA'$ such that $\phi^A = \phi^{A'} = \rho$,
define $X = (\1\ox K)\ket{\phi} < A\ox B$ and the sets of states,
\[
  \cS_{K,\rho} := \bigl\{ (\id\ox\cN)\phi : \cK(\cN) < K \bigr\}
                = \bigl\{ \sigma \in \cS(AB) : \operatorname{supp}\sigma < X,\ \sigma^A=\rho \bigr\},
\]
as well as, for $\epsilon > 0$,
\[
  \cS^{(\epsilon)}_{K,\rho}
  = \bigl\{ \sigma \in \cS(AB) : \exists \sigma'\in\cS_{K,\rho}
                                 \ \text{s.t.}\ \|\sigma-\sigma'\|_1 \leq \epsilon \bigr\}.
\]

\begin{proposition}
  \label{prop:compound}
  For any non-commutative bipartite graph $K < \cL(A\rightarrow B)$,
  a test state $\rho$ on $A$, and parameters $\epsilon > 0$ and an integer $k$,
  consider the family of cq-channels
  $\bigl[ W^\sigma:S_k \rightarrow \cS(A^k\ox B^k) : \sigma\in\cS^{(\epsilon)}_{K,\rho} \bigr]$,
  with
  \[
    W^\sigma : \pi \longmapsto (\1\ox U_\pi)\sigma^{\ox k}(\1\ox U_\pi)^\dagger.
  \]
  Then, for sufficiently large $\ell$, there is an $\ell$-block code
  of $N=2^{nR}$ messages ($n=k\ell$) and decoding POVM $(D_i)_{i=1}^N$, with rate
  \[
    R \geq \min_{\sigma\in\cS_{K,\rho}} I(A:B)_\sigma - 2\delta,
  \]
  and uniformly bounded error probability
  \[
    P_{\text{err}}\bigl((W^\sigma)^{\ox\ell}\bigr) 
        = \frac{1}{N}\sum_{i=1}^N 
            \Bigl( 1 - \tr\bigl( W^\sigma_{\pi_1(i)}\ox\cdots\ox W^\sigma_{\pi_\ell(i)} \bigr) D_i \Bigr)
        \leq c^\ell
  \]
  for all $\sigma\in\cS^{(\epsilon)}_{K,\rho}$.
  Here, $c < 1$ and 
  $\delta = 2\epsilon\log(|A||B|) + \frac{3}{k} + 2|B|^2 \frac{\log(k+|B|)}{k}$.
\end{proposition}
\begin{proof}
The family of cq-channels
$\bigl[ W^\sigma:S_k \rightarrow \cS(A^k\ox B^k) : \sigma\in\cS^{(\epsilon)}_{K,\rho} \bigr]$
generates a compound channel, meaning that on block length $\ell$, the
communicating parties face one of the i.i.d.~channels
$(W^\sigma)^{\ox \ell}$, $\sigma\in\cS^{(\epsilon)}_{K,\rho}$, but they
do not know beforehand which one, so they need to use a code that is
good for all of them.

For this we invoke the general result of Bjelakovic and 
Boche~\cite{quantum-compound-channel},
which states that there are such codes with rate 
\[
  \min_{\sigma\in\cS^{(\epsilon)}_{K,\rho}} 
    \chi\left(\left\{ p_\pi = \frac{1}{k!}, 
                W^\sigma_\pi = (\1\ox U_\pi)\sigma^{\ox k}(\1\ox U_\pi)^\dagger \right\}\right) - k\delta
\]
for any $\delta > 0$ and with error probability
uniformly bounded by $c^\ell$, $c=c(\delta)<1$.

By Lemma~\ref{lemma:shor} below, 
\[
  \chi\left(\left\{ p_\pi = \frac{1}{k!}, 
                    W^\sigma_\pi = (\1\ox U_\pi)\sigma^{\ox k}(\1\ox U_\pi)^\dagger \right\}\right) 
                                                              \geq k\,I(A:B)_\sigma - 2|B|^2 \log(k+|B|), 
\]
and because there is $\sigma'\in\cS_{K,\rho}$ with
$\|\sigma-\sigma'\|_1 \leq \epsilon$, Fannes' inequality~\cite{Fannes}
shows that the rate (over $n=k\ell$) is
\[
  \geq \min_{\sigma\in\cS_{K,\rho}} I(A:B)_\sigma 
         - 2\epsilon\log(|A||B|) - \frac{3}{k} - 2|B|^2 \frac{\log(k+|B|)}{k} - \delta,
\]
and we are done, choosing $\delta$ as advertised.

We end this proof pointing out a rather nice feature of the code:
each message is encoded as an $\ell$-tuple of permutations from $S_k$,
$i\mapsto \underline{\pi}(i) = \pi_1(i)\ldots\pi_\ell(i)$, which we may view naturally as
an element of $S_k \times \cdots \times S_k \subset S_n$, acting
on $B^n$ by permuting the tensor factors, each $\pi_j(i)$ on its own
block of $k$, hence message $i$ is mapped to the state
$W^\sigma_{\underline{\pi}(i)} 
= (\1\ox U_{\underline{\pi}(i)})\sigma^{\ox n}(\1\ox U_{\underline{\pi}(i)})^\dagger$
on $A^n B^n$.
\end{proof}

\begin{lemma}[Cf.~Shor~\cite{Shor:C_E}]
  \label{lemma:shor}
  For any channel $\cN:\cL(A)\rightarrow\cL(B)$ and a state $\rho$ 
  on $A$ with purification $\ket{\phi}\in AA'$, 
  and let $\sigma^{AB} = (\id\ox\cN)\phi$. 
  Then, for any integer $k$,
  \[
    \chi\left(\left\{ p_\pi = \frac{1}{k!}, 
                      W^\sigma_\pi = (\1\ox U_\pi)\sigma^{\ox k}(\1\ox U_\pi)^\dagger \right\}\right) 
                                                              \geq k\,I(A:B)_\sigma - 2|B|^2 \log(k+|B|), 
  \]
  where $\pi$ ranges over the symmetric group $S_k$, acting on $B^k$
  by permuting the tensor factors. 
\end{lemma}
\begin{proof}
With the average state 
\[
  \Omega^{A^k B^k} = \frac{1}{k!} \sum_{\pi\in S_k} (\1\ox U_\pi)\sigma^{\ox k}(\1\ox U_\pi)^\dagger,
\]
we have
\[\begin{split}
  \chi\left(\left\{ \frac{1}{k!}, (\1\ox U_\pi)\sigma^{\ox k}(\1\ox U_\pi)^\dagger \right\}\right) 
            &= S\bigl( \Omega^{A^k B^k} \bigr) - S\bigl( \sigma^{\ox k} \bigr)                   \\
            &= S\bigl( \Omega^{A^k} \bigr) + S\bigl( \Omega^{B^k} \bigr) - I(A^k:B^k)_\Omega 
                                                                 - S\bigl( \sigma^{\ox k} \bigr) \\
            &= k I(A:B)_\sigma - I(A^k:B^k)_\Omega,
\end{split}\]
where we have used that all ensemble members are just unitary transformed
versions of $\sigma^{\ox k}$ (first line), the definition of the mutual
information (second line), the fact that $\Omega^{A^k} = (\sigma^A)^{\ox k}$
and $\Omega^{B^k} = (\sigma^B)^{\ox k}$ as well as additivity of the von
Neumann entropy (third line).

Now we use the representation theory of $S_k$ acting on $B^k$ to bound the
mutual information remaining: From Schur-Weyl duality~\cite{FH1991} it is 
known that
\[
  B^k = \bigoplus_\lambda Q_\lambda^b \ox P_\lambda,
\]
where $\lambda$ are Young diagrams with at most $b=|B|$ rows, $P_\lambda$
are the corresponding irreps of $S_k$ and $Q_\lambda^b$ is the multiplicity
space, which is an irrep of the commutant representation, ${\rm SU}(b)$.
With the maximally mixed state $\tau_\lambda$ on $P_\lambda$, Schur's Lemma
implies that
\[
  \Omega^{A^k B^k} = \bigoplus_\lambda q_\lambda \omega_\lambda^{A^k Q_\lambda^b} \ox \tau_\lambda^{P_\lambda}.
\]
Now observe that $\Omega^{A^k B^k}$ can by local operations $B^k \leftrightarrow D := \bigoplus_\lambda Q_\lambda^b$ be reversibly transformed into
\[
  \widetilde{\Omega}^{A^k D} = \bigoplus_\lambda q_\lambda \omega_\lambda^{A^k Q_\lambda^b},
\]
hence
\[
  I(A^k:B^k)_\Omega = I(A^k:D)_{\widetilde{\Omega}} \leq 2 \log |D| \leq 2 b^2 \log(k+b).
\]
The latter because it is known that there are only $L\leq (k+1)^b$ 
Young diagrams and each ${\rm SU}(b)$ irrep has dimension
$|Q_\lambda^b| \leq M = (k+b)^{\frac12 b^2}$, hence
$|D| \leq LM = (k+1)^b (k+b)^{\frac12 b^2} \leq (k+b)^{b^2}$,
as we only need to consider the case $b\geq 2$.
\end{proof}

\medskip
\begin{proof}{\bf (of Theorem~\ref{thm:adversarial})}
First we show the upper bound, to be precise the strong 
converse. Because among the eligible channels is ${\cal N}^{\otimes n}$
with $\cK({\cal N}) < K$ attaining the minimum in $C_{\min E}(K)$, 
we see immediately that $C_{\ast E}(K) \leq C_E(\cN) = C_{\min E}(K)$.
In fact, the Quantum Reverse Shannon Theorem for 
$\cN^{\ox n}$~\cite{QRST,BertaChristandlRenner} implies the strong converse
as well, i.e.~for all $\epsilon < 1$,
\[
  \limsup_{n\rightarrow\infty} \frac1n N(n,\epsilon;K) \leq C_E(\cN) = C_{\min E}(K).
\]
A direct proof of this can be found in~\cite{GuptaWilde2013} (see also \cite{CooneyMosonyiWilde2014}).
Furthermore, Bowen~\cite{Bowen:feedback} (alternatively again the 
Quantum Reverse Shannon Theorem) showed that feedback does not increase
the entanglement-assisted capacity.

It remains to show achievability of $C_{\min E}(K)$;
for this it will be enough to show that for any test state $\rho$ 
on $A$, $C_{\ast E}(K) \geq \min_{\cK(\cN)<K} I(\rho;\cN)$, by 
exhibiting a sequence of codes with this rate and error probability 
going to $0$, exponentially in $n$.
Choose a purification $\ket{\phi}^{AA'}$ of $\rho$ and let
Alice and Bob share $\phi^{\ox n}$ as well as a maximally entangled 
state of Schmidt rank $n!$, which is measured by both parties
in the computational basis to obtain a shared random permutation
$\tau \in S_n$.
Alice's encoding will be to subject her $n$ input $A'$-systems to
a permutation $\underline{\pi}(i)$ for each message $i=1,\ldots,N$,
then apply $\tau$ and send the resulting state through the
channel $\cN^{(n)}$; Bob will apply the permutation $\tau^{-1}$ to his 
$n$ output $B$-systems. The state this prepares for Bob is
\[\begin{split}
  {\omega(i)}^{A^nB^n}
    &= \frac{1}{n!} \sum_{\tau\in S_n} 
                       (\1\ox U_\tau)^\dagger
                          \Bigl[ (\id\ox\cN^{(n)})
                               \big( (\1\ox U_\tau U_{\underline{\pi}(i)})
                                      \phi^{\ox n} 
                                     (\1\ox U_\tau U_{\underline{\pi}(i)})^\dagger \bigr)
                          \Bigr] (\1\ox U_\tau)                \\
    &= (\1\ox U_{\underline{\pi}(i)})
           \bigl[(\id\ox\overline{\cN}^{(n)})\phi^{\ox n}\bigr]
       (\1\ox U_{\underline{\pi}(i)})^\dagger                  \\
    &=: (\1\ox U_{\underline{\pi}(i)}) \sigma^{(n)} (\1\ox U_{\underline{\pi}(i)})^\dagger,
\end{split}\]
with the permutation-symmetrized channel
\[
  \overline{\cN}^{(n)}(\rho) 
     = \frac{1}{n!} \sum_{\tau\in S_n} 
                       U_\tau^\dagger 
                         \cN^{(n)}\big( U_\tau \rho U_\tau^\dagger \bigr)
                       U_\tau.
\]
Note that as $\cK(\cN^{(n)}) < K^{\ox n}$, the same holds for
$\overline{\cN}^{(n)}$. The permutations $\underline{\pi}(i)$
form a code for the compound channel 
\[
  \Bigl[
    W^\sigma_\pi = (\1\ox U_\pi)\sigma^{\ox k}(\1\ox U_\pi)^\dagger
     :
    \sigma^{AB} \in \cS^{(\epsilon)}_{K,\rho}
  \Bigr]
\]
according to Proposition~\ref{prop:compound} and its proof;
here, $n=k\ell$, and we will determine $k$ and $\epsilon$ later.
Bob will use the very decoding POVM $(D_i)$ from the same
proposition.

To analyze the performance of this strategy, we apply the 
Constrained Postselection Lemma~\ref{lemma:CPSL} to the 
permutation-symmetric state
$\sigma^{(n)} = (\id\ox\overline{\cN}^{(n)})\phi^{\ox n}$,
$X = (\1\ox K)\ket{\phi} < A \ox B$ and $\cR = \tr_B$:
\[
  \sigma^{(n)} \leq (n+1)^{3|A|^2|B|^2} 
                    \int {\rm d}\sigma\, \sigma^{\ox n}\, F(\sigma^A,\rho^A)^{2n},
\]
where the integral is over states $\sigma^{AB}$ supported on $X < AB$.
We split the integral into two parts, a first where
$F(\sigma^A,\rho^A) < 1-\alpha$ and a second one where
$F(\sigma^A,\rho^A) \geq 1-\alpha$. 
Choosing $\alpha$ small enough ensures that those $\sigma^{AB}$ 
are in $\cS^{(\epsilon)}_{K,\rho}$.
Thus, 
\[
  \sigma^{(n)} \leq (n+1)^{3|A|^2|B|^2} (1-\alpha)^{2n} \sigma_0 +
                    (n+1)^{3|A|^2|B|^2} 
                    \int_{F(\sigma^A,\rho^A) \geq 1-\alpha} {\rm d}\sigma\, \sigma^{\ox n},
\]
with some state $\sigma_0$. At this point we can evaluate the error probability:
\[\begin{split}
  P_{\text{err}} &= \frac{1}{N} \sum_{i=1}^N 
                       \tr\bigl( (\1\ox U_{\underline{\pi}(i)}) 
                                      \sigma^{(n)}
                                 (\1\ox U_{\underline{\pi}(i)})^\dagger (\1-D_i) \bigr) \\
                 &\leq \text{poly}(n) \left[ (1-\alpha)^{2n}
                                             + \max_{\sigma\in\cS^{(\epsilon)}_{K,\rho}}
                                                 \tr\bigl( (\1\ox U_{\underline{\pi}(i)}) 
                                                             \sigma^{(n)}
                                                           (\1\ox U_{\underline{\pi}(i)})^\dagger
                                                             (\1-D_i) \bigr)
                                       \right]                                          \\
                 &\leq \text{poly}(n)\bigl( (1-\alpha)^{2n} + c^{n/k} \bigr),
\end{split}\]
showing that for every $n$ and $\epsilon$ the error probability goes
to zero exponentially -- in fact, at the same rate as the corresponding 
compound channel, except for the additional term $(1-\alpha)^{2n}$.

The rate, according to Proposition~\ref{prop:compound} is
$\geq \min_{\cK(\cN)<K} I(\rho;\cN) - 2\delta$, where
$\delta = 2\epsilon\log(|A||B|) + \frac{3}{k} + 2|B|^2 \frac{\log(k+|B|)}{k}$
can be made arbitrarily small by choosing $\epsilon$ small
enough and $k$ large enough.
\end{proof}

\medskip
\begin{remark}
Along the same lines, the use of permutation-symmetrization and 
the Postselection Lemma allow to give a new proof of the coding theorem 
for \emph{arbitrarily varying cq-channels}~\cite{AVQC-C}, by reducing
it to a compound cq-channel~\cite{quantum-compound-channel}, 
cf.~also~\cite{env-assist}.

Observe however that what we treated here is not an ``arbitarily varying
quantum channel'' in any sense previously considered~\cite{AVQC,AVQC-C},
going beyond the model in~\cite{env-assist}, too.
\end{remark}

\section{A Constrained Post-Selection Lemma}
\label{app:postselection}
Here we show the following extension of the main technical result 
of~\cite{CKR-postselect} (albeit with a worse polynomial prefactor).
\begin{lemma}
  \label{lemma:CPSL}
  For given Hilbert space $X$ with dimension $d$, denote by 
  ${\rm d}\sigma$ the measure on the quantum states $\mathcal{S}(X)$ 
  obtained by drawing a pure state from $X \ox X'$ uniformly at
  random (i.e., from the unitarily invariant probability measure)
  and tracing out $X'$.
  
  Then, for any $S_n$-invariant state $\rho^{(n)}$ on $X^n$,
  \[
    \rho^{(n)} \leq (n+1)^{3d^2} \int {\rm d}\sigma\, \sigma^{\ox n}\, F\bigl( \rho^{(n)}, \sigma^{\ox n} \bigr)^2.
  \]
  The measure ${\rm d}\sigma$ is universal in the sense that it depends only on the space $X$.
  
  Furthermore, let ${\cal R}:{\cal L}(X) \rightarrow {\cal L}(Y)$ be a cptp map,
  $\eta \in \mathcal{S}(Y)$ a state. Then, for every $S_n$-invariant state $\rho^{(n)}$
  on $X^n$ with $\mathcal{R}^{\ox n}\bigl( \rho^{(n)} \bigr) = \eta^{\ox n}$,
  \[
    \rho^{(n)} \leq (n+1)^{3d^2} \int {\rm d}\sigma\, \sigma^{\ox n}\, F\bigl( \mathcal{R}(\sigma),\eta \bigr)^{2n}.
  \]
  Note that the right hand side depends only on $X$, $\mathcal{R}$, $\eta$ and $n$. 
\end{lemma}

\medskip
Here, $F(\xi,\eta) = \| \sqrt{\xi}\sqrt{\eta} \|_1$ is the
fidelity between (mixed) states $\xi,\eta\in\mathcal{S}(X)$
\cite{Uhlmann,Jozsa,Fuchs-vandeGraaf}.

\medskip
\begin{remark}
Note that in $\Omega^{(n)}$, the contribution of states $\sigma$
with $F\bigl( \mathcal{R}(\sigma),\eta \bigr) < 1-\epsilon$ is exponentially
small in $n$. I.e., for a symmetric state with an additional constraint,
expressed by $\mathcal{R}$ and $\eta$, the universal de Finetti state 
from~\cite{CKR-postselect} may be chosen in such a way that almost all
its contributions also approximately obey the constraint.
\end{remark}

\medskip
\begin{proof}
Denoting the uniform (i.e.~unitarily invariant) probability 
measure over \emph{pure} states $\zeta = \proj{\zeta}$
on $X\ox X'$ by ${\rm d}\zeta$, it is well known that
\[
  \int {\rm d}\zeta\, \zeta^{\otimes n} = \frac{1}{{n+d^2-1 \choose d^2-1}} \Pi_{\text{Sym}^n(X\ox X')},
\]
with $\Pi_{\text{Sym}^n(X\ox X')}$ denoting the projector onto the (Bose)
symmetric subspace of $(X\ox X')^{\otimes n}$. The reason is that the latter
is an irrep of the $U^{\otimes n}$-representation for $U \in \text{SU}(d^2)$, 
so Schur's Lemma applies.
Now we apply Caratheodory's Theorem, which says that ${\rm d}\zeta$ can be 
convex-decomposed into measures with finite support, more precisely
ensembles $\{q_i,\zeta_i\}_{i=1}^{D^2}$, with $D={n+d^2-1 \choose d^2-1} \leq (n+1)^{d^2}$,
the dimension of the Bose symmetric subspace of $(X\otimes X')^{\otimes n}$, 
and
\[
  \sum_i q_i \zeta_i^{\otimes n} = \frac{1}{D} \Pi_{\text{Sym}^n(X\ox X')}.
\]
For the moment we shall focus on one of these measures/ensembles.

It is also well known that one can purify $\rho^{(n)}$ in a Bose symmetric way,
i.e.~$\rho^{(n)} = \tr_{{X'}^n} \varphi^{(n)}$, with 
$\varphi^{(n)} = \proj{\varphi^{(n)}}$ a pure state supported on the Bose
symmetric subspace. Thus, with the operator $A := \sum_i \ket{\zeta_i}^{\ox n}\!\bra{i}$,
\begin{equation*}\begin{split}
  \varphi^{(n)} &= \Pi_{\text{Sym}^n(X\ox X')} \varphi^{(n)} \Pi_{\text{Sym}^n(X\ox X')} \\
                &= D^2 \sum_{ij} q_i q_j \zeta_i^{\ox n}\, \varphi^{(n)}\, \zeta_j^{\ox n} \\
                &= D^2 A \left( \sum_{ij} q_i q_j \ketbra{i}{j} 
                                  \bra{\zeta_i}^{\ox n} \varphi^{(n)} \ket{\zeta_j}^{\ox n} \right) 
                       A^\dagger \\
                &\leq D^4 A \left( \sum_{i} q_i^2 \ketbra{i}{i} 
                                     \bra{\zeta_i}^{\ox n} \varphi^{(n)} \ket{\zeta_i}^{\ox n} \right) 
                          A^\dagger \\
                &\leq D^3 A \left( \sum_{i} q_i \proj{i}
                                     F\left( \zeta_i^{\ox n}, \varphi^{(n)}\right)^2 \right) 
                          A^\dagger \\
                &\leq D^3 \sum_i q_i \zeta_i^{\ox n}\, 
                                       F\left( (\tr_{X'}\zeta_i)^{\ox n}, \rho^{(n)} \right)^2, 
\end{split}\end{equation*}
where in the fourth line we have used Hayashi's pinching inequality
\cite{Hayashi-pinch}, 
and in the fifth $q_i \leq \frac{1}{D}$; 
in line six 
we have invoked the monotonicity
of the fidelity under cptp maps, here the partial trace,
as well as $\tr_{{X'}^n} \varphi^{(n)} = \rho^{(n)}$. 

Now we remember that $\{q_i,\zeta_i\}$ was just one of the Caratheodory components of
the uniform measure ${\rm d}\zeta$, so by convex combination,
\[
  \varphi^{(n)} 
      \leq D^3 \int {\rm d}\zeta\, \zeta^{\ox n}\, F\left( (\tr_{X'}\zeta)^{\ox n}, \rho^{(n)} \right)^2,
\]
hence by partial trace over ${X'}^n$, and recalling the definition
of ${\rm d}\sigma$, we arrive at
\[
  \rho^{(n)} \leq D^3 \int {\rm d}\sigma\, \sigma^{\ox n}\, F\bigl( \sigma^{\ox n}, \rho^{(n)} \bigr)^2.
\]

To obtain the second bound, we apply the map $\mathcal{R}^{\ox n}$ to the
states inside the above fidelity; by monotonicity of the fidelity
once more,
\[\begin{split}
  F\bigl( \sigma^{\ox n}, \rho^{(n)} \bigr) 
         &\leq F\left( \mathcal{R}^{\ox n}(\sigma^{\ox n}), \mathcal{R}^{\ox n}(\rho^{(n)}) \right) \\
         &=    F\left( \bigl(\mathcal{R}(\sigma)\bigr)^{\ox n}, \eta^{\ox n} \right) \\
         &=    F\bigl( \mathcal{R}(\sigma), \eta)^n,
\end{split}\]
as desired.
\end{proof}

\medskip
\begin{remark}
  It is the trick to sandwich the Bose-symmetric state $\varphi^{(n)}$ between
  symmetric subspace projectors -- rather than bounding it directly by
  that projector --, which allows the introduction of fidelities
  between the state and ``test'' product states.
  
  Here we have used this to enforce a linear constraint valid for $\rho^{(n)}$
  on the components of the de Finetti state on the right hand side. It turns 
  out, perhaps unsurprisingly, that also other convex constraints
  (with a ``good'' behaviour linking $n=1$ with the general case) are
  amenable to the same treatment, for instance membership in the convex
  set of separable states for a multipartite space $X = X_1\ox \cdots \ox X_k$,
  and other similar sets, or even non-convex constraints. Such
  generalizations and their applications are discussed in~\cite{LancienWinter2015}.
\end{remark}


\begin{thebibliography}{99}
\bibitem{AdamiCerf} Chris Adami, Nicolas J. Cerf, ``Von Neumann capacity of noisy
  quantum channels'',
  Phys. Rev. A {\bf 56}(5):3470-3483 (1997).

\bibitem{Ahlswede:feedback} Rudolf Ahlswede, ``Channels with Arbitrarily Varying
  Channel Probability Functions in the Presence of Noiseless Feedback'',
  Z. Wahrsch. Verw. Geb. {\bf 25}:239-252 (1973).

\bibitem{AVQC} Rudolf Ahlswede, Igor Bjelakovic, Holger Boche and Janis N\"otzel,
  ``Quantum Capacity under Adversarial Quantum Noise:
  Arbitrarily Varying Quantum Channels'', 
  Commun. Math. Phys. {\bf 317}(1):103-156 (2013).

\bibitem{BJOP} Somshubhro Bandyopadhyay, Rahul Jain, Jonathan Oppenheim and 
  Christopher Perry,
  ``Conclusive Exclusion of Quantum States'',
  Phys. Rev. A {\bf 89}:022336 (2014).

\bibitem{QRST} Charles H. Bennett, Igor Devetak, Aram W. Harrow, Peter W. Shor
  and Andreas Winter,
  ``The Quantum Reverse Shannon Theorem and Resource Tradeoffs for Simulating
  Quantum Channels'',
  IEEE Trans. Inf. Theory {\bf 60}(5):2926-2959 (2014).

\bibitem{BSST} Charles H. Bennett, Peter W. Shor, John A. Smolin and Ashish V. Thapliyal,
  ``Entanglement-assisted classical capacity of noisy quantum channels'',
  Phys. Rev. Lett. {\bf 83}(15):3081-3084 (1999);
  ``Entanglement-assisted capacity of a quantum channel and the reverse 
  Shannon theorem'', IEEE Trans. Inf. Theory {\bf 46}(10):2637-2655 (2002).
  
\bibitem{BertaChristandlRenner} Mario Berta, Matthias Christandl and Renato Renner,
  ``A Conceptually Simple Proof of the Quantum Reverse Shannon Theorem'',
  in: \emph{Proc. TQC 2010}, LNCS 6519, pp. 131-140, 
  Springer Verlag, Berlin Heidelberg New York, 2011;
  ``The Quantum Reverse Shannon Theorem Based on One-Shot Information Theory'',
  Commun. Math. Phys. {\bf 306}(3):579-615 (2011).

\bibitem{quantum-compound-channel} Igor Bjelakovic and Holger Boche,
  ``Classical Capacities of Compound and Averaged Quantum Channels'',
  IEEE Trans. Inf. Theory {\bf 55}(7):3360-3374 (2009).

\bibitem{AVQC-C} Igor Bjelakovic, Holger Boche, Giesbert Jan\ss{}en and Janis N\"otzel,
  ``Arbitrarily Varying and Compound Classical-Quantum Channels
  and a Note on Quantum Zero-Error Capacities'',
  in: \emph{Information Theory, Combinatorics, and Search Theory: In Memory of
  Rudolf Ahlswede}, H. Aydinian, F. Cicalese, and C. Deppe (eds.),
  LNCS 7777, pp. 247-283, Springer Verlag, Berlin Heidelberg New York, 2013;
  arXiv[quant-ph]:1209.6325.

\bibitem{Bowen:feedback} Garry Bowen, ``Quantum Feedback Channels'',
  IEEE Trans. Inf. Theory {\bf 50}(10):2429-2433 (2004).


\bibitem{CJW-pure-trans} Anthony Chefles, Richard Jozsa and Andreas Winter, 
  ``On the existence of physical transformations
  between sets of quantum states'', 
  Int. J. Quantum Inf. {\bf 2}(1):11-21 (2004).

\bibitem{Christandl2006} Matthias Christandl, \emph{The Structure of Bipartite Quantum 
  States --- Insights from Group Theory and Cryptography}, 
  PhD thesis, University of Cambridge, 2006. arXiv:quant-ph/0604183.

\bibitem{CKR-postselect} Matthias Christandl, Robert K\"onig and Renato Renner,
  ``Postselection Technique for Quantum Channels with Applications
  to Quantum Cryptography'',
  Phys. Rev. Lett. {\bf 102}:020503 (2009).

\bibitem{CoverThomas} Thomas M. Cover and Joy A. Thomas,
  \emph{Elements of Information Theory} (2nd edition),
  Wiley \&{} Sons, New York, 2006.

\bibitem{Chvatal:LP} V\'{a}clav Chv\'{a}tal, \emph{Linear Programming},
  W. H. Freeman, New York, 1983.

\bibitem{Csiszar:relent} Imre Csisz\'{a}r, ``Generalized Cutoff Rates and R\'{e}nyi's 
  Information Measures'',
  IEEE Trans. Inf. Theory {\bf 41}(1):26-34 (1995).

\bibitem{Csiszar:method-of-types} Imre Csisz\'{a}r, ``The Method of Types'',
  IEEE Trans. Inf. Theory {\bf 44}(6):2505-2523 (1998).

\bibitem{CCH} Toby S. Cubitt, Jianxin Chen and Aram W. Harrow, 
  ``Superactivation of the Asymptotic Zero-Error Classical Capacity 
  of a Quantum Channel'',
  IEEE Trans. Inf. Theory {\bf 57}(12):8114-8126 (2011).

\bibitem{CLMW-ent-zero} Toby S. Cubitt, William Matthews, Debbie Leung and Andreas Winter,
  ``Improving Zero-Error Classical Communication with Entanglement'',
  Phys. Rev. Lett. {\bf 104}:230503 (2010).

\bibitem{CLMW:0} Toby S. Cubitt, William Matthews, Debbie Leung and Andreas Winter,
  ``Zero-error channel capacity and simulation assisted by non-local correlations'',
  IEEE Trans. Inf. Theory {\bf 57}(8):5509-5523 (2011).

\bibitem{HsiehDattaWilde} Nilanjana Datta, Min-Hsiu Hsieh and Mark M. Wilde,
  ``Quantum Rate Distortion, Reverse Shannon Theorems, and Source-Channel Separation",
   IEEE Trans. Inf. Theory {\bf 59}(1):615-630 (2013).



\bibitem{Duan:zero} Runyao Duan, ``Super-Activation of Zero-Error Capacity of Noisy
  Quantum Channels'', arXiv[quant-ph]:0906.2527 (2009).

\bibitem{DSW:q-theta} Runyao Duan, Simone Severini and Andreas Winter, 
  ``Zero-error communication
  via quantum channels, non-commutative graphs and a quantum Lov\'{a}sz number'', 
  IEEE Trans. Inf. Theory {\bf 59}(2):1164-1174 (2013).

\bibitem{DW:ns-ass} Runyao Duan and Andreas Winter,
  ``Non-Signalling Assisted Zero-Error Capacity of Quantum Channels and an
  Information Theoretic Interpretation of the Lov\'{a}sz Number'',
  arXiv[quant-ph]:1409.3426 (2014).

\bibitem{DuanWang} Runyao Duan and Xin Wang,
  ``Activated zero-error classical capacity of quantum channels in
  the presence of quantum no-signalling correlations'',
  arXiv[quant-ph]:1510.05437 (2015).

\bibitem{GuptaWilde2013} Manish K. Gupta and Mark M. Wilde, 
  ``Multiplicativity of completely bounded $p$-norms implies a strong converse 
  for entanglement-assisted capacity'', arXiv[quant-ph]:1310.7028 (2013).
  
\bibitem{CooneyMosonyiWilde2014} Tom Cooney, Mil\'an Mosonyi and Mark M. Wilde,
  ``Strong converse exponents for a quantum channel discrimination problem and 
  quantum-feedback-assisted communication", 
  arXiv[quant-ph]:1408.3373 (2014).
  
\bibitem{Elias:list} Peter Elias, ``Zero Error Capacity Under List Decoding'',
  IEEE Trans. Inf. Theory {\bf 34}(5):1070-1074 (1987).

\bibitem{Fannes} Mark Fannes, ``A Continuity Property of the Entropy Density
  for Spin Lattice Systems'',
  Commun. Math. Phys. {\bf 31}:291-294 (1973).

\bibitem{Fuchs-vandeGraaf} Christopher A. Fuchs and Jeroen van de Graaf,
  ``Cryptographic Distinguishability Measures for Quantum-Mechanical States'',
  IEEE Trans. Inf. Theory {\bf 45}(4):1216-1227 (1999).

\bibitem{FH1991} William Fulton and Joe Harris, \emph{Representation Theory: A First Course},
   Springer Verlag, Berlin Heidelberg New York, 1991.

\bibitem{Harrow2005} Aram W. Harrow, \emph{Applications of Coherent Classical Communication
  and the Schur transform to quantum information theory}, PhD thesis, MIT, 2005.

\bibitem{HausladenJozsaSchumacherWestmorelandWootters} Paul Hausladen, Richard Jozsa,
  Benjamin Schumacher, Michael D. Westmoreland and William K. Wootters, 
  ``Classical information capacity of a quantum channel'',
  Phys. Rev. A {\bf 54}(3):1869-1876 (1996).

\bibitem{Hayashi-pinch} Masahito Hayashi, ``Optimal sequence of POVMs in the sense
  of Stein's lemma in quantum hypothesis testing'',
  J. Phys. A: Math. Gen. {\bf 35}(5):10759Ð10773 (2002).
  Tomohiro Ogawa and Masahito Hayashi, ``On Error Exponents in Quantum Hypothesis Testing'',
  IEEE Trans. Inf. {\bf 50}(6):1368-1372 (2004).

\bibitem{Jozsa} Richard Jozsa, ``Fidelity for mixed quantum states'', 
  J. Mod. Opt. {\bf 41}(12):2315-2323 (1994).


\bibitem{env-assist} Siddharth Karumanchi, Stefano Mancini, Andreas Winter and Dong Yang, 
  ``Quantum Channel Capacities with Passive Environment Assistance'',
  arXiv[quant-ph]:1407.8160 (2014).

\bibitem{KoernerOrlitsky} J\'{a}nos K\"orner and Alon Orlitsky,
  ``Zero-Error Information Theory'',
  IEEE Trans. Inf. Theory {\bf 44}(6):2207-2229 (1998).

\bibitem{LancienWinter2015} C\'{e}cilia Lancien and Andreas Winter,
  ``Parallel repetition and concentration for (sub-)no-signalling games via a 
  flexible constrained de Finetti reduction",
  arXiv[quant-ph]:1506.07002 (2015). 

\bibitem{IQC-ent-zero} Debbie Leung, Laura Man\v{c}inska, William Matthews,
  Maris Ozols and Aidan Roy,
  ``Entanglement can Increase Asymptotic Rates of Zero-Error Classical 
  Communication over Classical Channels'',
  Commun. Math Phys. {\bf 311}(1):97-111 (2012).

\bibitem{prehist} Rex A. C. Medeiros, Romain Alleaume, G\'{e}rard Cohen and Francisco M. de Assis,
  ``Zero- error capacity of quantum channels and noiseless subsystems'' 
  in: Proc. Int. Telecommun. Symp., Fortaleza, CE, Brazil, 3-6 September 2006, pp. 900-905;
  ``Quantum states characterization for the zero-error capacity'',
  arXiv:quant-ph/0611042 (2006).

\bibitem{PBR} Matthew F. Pusey, Jonathan Barrett and Terry Rudolph, 
  ``On the reality of the quantum state'',
  Nature Phys. {\bf 8}(6):475-478 (2012); arXiv[quant-ph]:1111.3328.

\bibitem{ScheinermanUllman} E. R. Scheinerman, D. H. Ullman, 
  \emph{Fractional Graph Theory: A Rational Approach to the Theory of Graphs}, 
  Vol. 46 of Wiley Series in Discrete Mathematics and Optimization, 
  John Wiley \&{} Sons, 1997.

\bibitem{SW:optimal-ens} Benjamin Schumacher and Michael D. Westmoreland,
  ``Optimal signal ensembles'', Phys. Rev. A {\bf 63}:022308 (2001).

\bibitem{Shannon48} Claude E. Shannon, ``A mathematical theory of communication'',
  Bell Syst. Tech. J. {\bf 27}:379-423 \&{} 623-656 (1948).

\bibitem{Shannon56} Claude E. Shannon, ``The zero-error capacity of a noisy channel'',
  IRE Trans. Inf. Theory {\bf 2}:8-19 (1956).

\bibitem{Shor:C_E} Peter W. Shor, 
  ``The classical capacity achievable by a quantum channel assisted by limited entanglement'',
  Quantum Inf. Comput. {\bf 4}(6 \&{}7):537-545 (2004);
  arXiv:quant-ph/0402129 (2004).

\bibitem{Sion} Maurice Sion, ``On general minimax theorems'',
  Pacific J. Math. {\bf 8}(1):171-176 (1958).

\bibitem{Stinespring} W. Forrest Stinespring,
  ``Positive Functions on $C^\ast$-Algebras'',
  Proc. Amer. Math. Soc. {\bf 6}(2):211-216 (1955).

\bibitem{TKG:unambiguous} Masahiro Takeoka, Hari Krovi and Saikat Guha,
  ``Achieving the Holevo Capacity of a Pure State Classical-Quantum Channel 
  via Unambiguous State Discrimination'',
  Proc. ISIT 2013, pp. 166-170 (2013).

\bibitem{Uhlmann} Armin Uhlmann, ``The `transition probability' in the state space
  of a $\ast$-algebra'', 
  Reps. Math. Phys. {\bf 9}:273-279 (1976).
  
\end{thebibliography}
\end{document}